\pdfoutput=1
\documentclass[preprint,3p,authoryear,times]{elsarticle}
\usepackage{graphicx}
\usepackage{amssymb}
\usepackage{mathtools}
\usepackage{amsthm}
\newtheorem{thm}{Theorem}
\theoremstyle{remark}
\newtheorem{rmk}{Remark}
\newbox\JIGAMMa
\newbox\JIGAMMb
\newbox\JIGAMMc
\newbox\JIGAMMd
\newbox\JIGAMMz
\newbox\LLbox
\newbox\LLboxh
\newbox\SLhilfbox
\newbox\SLubox
\newbox\SLobox
\newbox\SLergebnis
\newbox\TENbox
\newif\ifSLoben
\newif\ifSLunten
\newdimen\JIGAMMdimen
\newdimen\JIhsize\relax\JIhsize=\hsize
\newdimen\SLrandausgleich
\newdimen\SLhoehe
\newdimen\SLeffbreite
\newdimen\SLuvorschub
\newdimen\SLmvorschub
\newdimen\SLovorschub
\newdimen\SLsp
\def\pkt{\cdot}
\def\ppkt{\mathbin{\mathord{\cdot}\mathord{\cdot}}}
\setbox\JIGAMMa = \hbox{$\scriptscriptstyle c$} \setbox\JIGAMMz =
\hbox{\hskip-.35pt\vrule width .25pt\hskip-.35pt
                        \vbox to1.2\ht\JIGAMMa{\vskip-.125pt
                             \hrule width1.2\ht\JIGAMMa height.25pt
                             \vfill
                             \hrule width1.2\ht\JIGAMMa height.25pt
                             \vskip-.125pt}%
                        \hskip-.125pt\vrule width .25pt\hskip-.125pt}
\def\Oldroy#1#2#3{\STAPEL{#1}!_\SLstrich!_\SLstrich!^\circ{}
                  \ifx #2,{}_{\copy\JIGAMMz}%
                  \else \mskip1mu{}^{\copy\JIGAMMz}\fi
                  \mskip1mu\ifx #3,{}_{\copy\JIGAMMz}%
                           \else {}^{\copy\JIGAMMz}\fi }
\def\OP#1#2{\ifnum#1=1{\rm S}
            \else\ifnum#1=2{\rm S}^\JIv
                 \else\ifnum#1=3{\rm S}^T
                      \else{{\rm S}^T}^\JIv
            \fi\fi\fi\LL{{#2}}\RR}
\edef\JIminus{{\setbox\JIGAMMa=\hbox{$\scriptstyle x$}%
           \hbox{\hskip .10\wd\JIGAMMa
                 \vbox{\hrule width .6\wd\JIGAMMa height .07\wd\JIGAMMa
                       \vskip.53\ht\JIGAMMa}%
                 \hskip .10\wd\JIGAMMa}}}
\edef\JIv{{\JIminus 1}}

\def\LL#1\RR{\setbox\LLbox =\hbox{\mathsurround=0pt$\displaystyle
                                              \left(#1\right)$}%
       \setbox\LLboxh=\hbox{\mathsurround=0pt%
                  $\displaystyle{\left(%
                      \vrule width 0pt height\ht\LLbox depth\dp\LLbox
                      \right)}$}%
       \left(\hskip-.3\wd\LLboxh\relax\copy\LLbox
              \hskip-.3\wd\LLboxh\relax\right)}
\def\ZBOX#1#2#3{\def#3{}%
                \setbox#1 = #2
                \def#3{ to \wd#1}%
                \setbox#1 = #2}
\def\SLdreieck{\setbox\TENbox=\hbox{\fontscsy\char 52}%\fontscex\char\dq 65
                 \dp\TENbox = 0pt
                 \hbox{\hskip -2\SLrandausgleich
                       \box\TENbox
                       \hskip -2\SLrandausgleich}}
\def\SLtilde{\setbox\TENbox=\hbox{\fontscex\char 101}%{\fontscex\char\dq 65
                    \vbox{\vskip-.03\ht\TENbox
                          \hbox{\hskip -1\SLrandausgleich
                                \copy\TENbox
                                \hskip -1\SLrandausgleich}
                          \vskip -.86\ht\TENbox}}
\def\SLstrich{\vrule width \SLeffbreite height.4pt}
\def\SLpunkt{{\vbox{\hbox{$\displaystyle.$}\vskip.03cm}}}
\def\SLabstand{\vskip .404pt}
\def\SLzwischen{\vskip 1.372pt}
\font\fontscsy=cmsy6 \font\fontscex=cmex10 scaled 1200
\def\STAPEL#1{\def\SLkern{#1}%
              \futurelet\next\SLpruef
               A_0   _0    :B_0   _-.17 :C_.05 _-.15 :D_0   _-.2
              :E_0   _-.2  :F_0   _-.21 :G_0   _-.15 :H_0   _-.23
              :I_.2  _.15  :J_.05 _-.1  :K_0   _-.22 :L_0   _-.1
              :M_0   _-.23 :N_0   _-.25 :O_.05 _-.2  :P_0   _-.21
              :Q_.05 _-.2  :R_0   _-.03 :S_.15   _-.15 :T_.2  _0
              :U_.1  _-.1  :V_.1  _-.15 :W_.1  _-.2  :X_0   _-.22
              :Y_.16 _-.15 :Z_0   _-.25
              :a_.05 _-.05 :b_.05 _0    :c_.05 _.05  :d_0   _-.05
              :e_.07 _0    :f_0   _-.15 :g_.04 _-.2  :h_0   _-.07
              :i_.05 _0    :j_.08 _-.1  :k_0   _-.1  :l_.2  _.15
              :m_0   _-.1  :n_0   _-.1  :o_0   _-.1  :p_.15 _0
              :q_.1  _0    :r_.1  _-.1  :s_0   _-.2  :t_.1  _.05
              :u_0   _-.1  :v_0   _-.2  :w_0   _-.2  :x_.04 _-.14
              :y_.15 _-.05 :z_0   _-.15
              :\mit\Phi_.08 _-.1    :\mit\Omega_0 _-.2   :\varXi_.00 _-.2
              :\alpha_0 _-.2        :\gamma_.1 _-.1      :\varepsilon_.1 _-.1
              :\epsilon_.05 _-.05   :\eta_.05 _-.15      :\lambda_0 _0
              :\mu_0 _-.25          :\nu_0 _-.2          :\varSigma_-.03 _-.2
              :\varrho_.00 _-.2     :\sigma_.1 _-.2      :\tau_.15 _-.15
              :\varphi_.2 _-.1      :\omega_.1 _-.1      :\mit\Gamma_-.1 _-.1
              :\Lambda_0 _0         :\Gam_0 _0           :\Lam_0 _0
              :\SLsuchende
              \def\SLtrick{\noexpand\SLtrick\noexpand}%
                \def\SLdummy{\noexpand\SLdummy}%
                \edef\SLoboxinhalt{}\edef\SLuboxinhalt{}%
                \SLobenfalse\SLuntenfalse
                \futurelet\next\SLsuchruf}
  \def\SLsuchruf{\ifx\next !\let\next\SLexpand
                 \else\let\next\SLerzeug\fi\next}
  \def\SLexpand#1#2#3{\ifx #2\sb\ifSLunten\let\SLspeicher\SLuboxinhalt
                   \else\def\SLspeicher{\SLtrick\SLabstand}\fi
                   \edef\SLuboxinhalt{%
                       \SLspeicher
                       \SLtrick\SLzwischen
                       \hbox\SLdummy{\hfil\mathsurround=0pt
$\SLtrick\scriptstyle\SLtrick#3$%
                                     \hfil}}%
                   \SLuntentrue%
                      \else\ifSLoben\let\SLspeicher\SLoboxinhalt
                   \else\def\SLspeicher{\SLtrick\SLabstand}\fi
                   \edef\SLoboxinhalt{%
                       \hbox\SLdummy{\hfil\mathsurround=0pt
$\SLtrick\scriptstyle\SLtrick#3$%
                                     \hfil}%
                       \SLtrick\SLzwischen
                       \SLspeicher}%
                   \SLobentrue\fi\futurelet\next\SLsuchruf}
  \def\SLerzeug{\def\SLtrick{}
                \setbox\SLhilfbox=\hbox{$\displaystyle{E}$}%
                \SLrandausgleich=.04\wd\SLhilfbox
                      \setbox\SLhilfbox=%
                         \hbox{\hskip -1\SLrandausgleich
                          \mathsurround=0pt$\displaystyle{\SLkern}$%
                               \hskip -1\SLrandausgleich}%
                      \SLhoehe = \ht\SLhilfbox
                      \advance\SLhoehe by \dp\SLhilfbox
                      \SLeffbreite = \wd\SLhilfbox
                      \advance\SLeffbreite by \SLab\SLhoehe
                      \ZBOX\SLubox{\vbox{\offinterlineskip
                                         \SLuboxinhalt
                                         \hrule height 0pt}}\SLdummy
                      \ZBOX\SLobox{\vbox{\offinterlineskip
                                         \SLoboxinhalt
                                         \hrule height 0pt}}\SLdummy
                      \SLsp = \SLzu\SLhoehe
                      \advance\SLsp by -.5\SLeffbreite
                      \SLuvorschub = -1\SLsp
                      \advance\SLuvorschub by -.5\wd\SLubox
                      \SLovorschub = -1\SLsp
                      \advance\SLovorschub by -.5\wd\SLobox
                      \advance\SLovorschub by .26\SLhoehe
                      \ifdim\SLuvorschub > \SLovorschub
                         \SLsp = \SLovorschub
                      \else
                         \SLsp = \SLuvorschub
                      \fi
                      \ifdim\SLsp < 0pt%
                         \advance\SLuvorschub by -1\SLsp
                         \SLmvorschub = -1\SLsp
                         \advance\SLovorschub by -1\SLsp
                      \else
                         \SLmvorschub = 0pt
                      \fi
                      \setbox\SLergebnis = \hbox{%
                         \offinterlineskip
                         \hskip\SLrandausgleich\relax
                         \vbox to 0pt{%
                            \vskip -1\ht\SLobox
                            \vskip -1\ht\SLhilfbox
\hbox{\hskip\SLovorschub\copy\SLobox\hfil}%
                            \hbox{\hskip\SLmvorschub\copy\SLhilfbox
                                  \hfil}%
\hbox{\hskip\SLuvorschub\copy\SLubox\hfil}%
                            \vss}%
                         \hskip\SLrandausgleich}%
                      \SLsp = \dp\SLhilfbox
                      \advance\SLsp by \ht\SLubox
                      \dp\SLergebnis = \SLsp
                      \SLsp = \ht\SLhilfbox
                      \advance\SLsp by \ht\SLobox
                      \ht\SLergebnis = \SLsp
                      \box\SLergebnis{}}
  \def\SLpruef{\ifx\next\SLsuchende\def\SLzu{0}\def\SLab{0}%
                  \def\next##1\SLsuchende{\relax}%
               \else\let\next\SLvergl
               \fi\next}
  \def\SLvergl#1_#2_#3:{\def\SLv{#1}%
                        \ifx\SLkern\SLv\def\SLzu{#2}\def\SLab{#3}%
\def\next##1\SLsuchende{\relax}%
                        \else\def\next{\futurelet\next\SLpruef}
                        \fi\next}
\newbox\minusbox
\def\minus{\mathchoice{\minusarb\displaystyle}%
                      {\minusarb\textstyle}%
                      {\minusarb\scriptstyle}%
                      {\minusarb\scriptscriptstyle}}
  \def\minusarb#1{\setbox\minusbox=\hbox{$#1x$}%
                  \hbox{\hskip .10\wd\minusbox
                        \vbox{\hrule width .6\wd\minusbox
                                     height .07\wd\minusbox
                              \vskip.53\ht\minusbox}%
                        \hskip .10\wd\minusbox}}

\def\ECK#1{#1!^\SLdreieck}
\def\Ttilde{\STAPEL T!_\SLstrich!_\SLstrich!^\SLtilde}

\def\Tenz#1{\STAPEL {#1}!_\SLstrich!_\SLstrich}
\def\Vek#1{\STAPEL {#1}!_\SLstrich}
\def\Tenv#1{\STAPEL {#1}!_\SLstrich!_\SLstrich!_\SLstrich!_\SLstrich}
\def\Ceck{\STAPEL C!_\SLstrich!_\SLstrich!^\SLdreieck}
\def\CG{\STAPEL C!_\SLstrich!_\SLstrich!^G}

\def\ABL#1\nach#2{\frac{\partial #1}{\partial #2}}
\def\ppp{\quad .}
\def\von#1{\left(#1\right)}
\def\PKTs#1{\STAPEL{#1}!^\SLpunkt}
\def\F{\STAPEL F!_\SLstrich!_\SLstrich}
\def\C{\STAPEL C!_\SLstrich!_\SLstrich}
\def\I{\STAPEL I!_\SLstrich!_\SLstrich}
\def\X{\STAPEL X!_\SLstrich!_\SLstrich}
\def\Y{\STAPEL Y!_\SLstrich!_\SLstrich}
\def\e{\STAPEL e!_\SLstrich}
\def\Sollsein{\stackrel{!}{=}}
\def\rhot{\widetilde\rho}
\def\Nab{\Vek\nabla}
\def\Nabtil{{\widetilde\Nab}}

\def\inv#1{#1^{\minus 1}}
\usepackage{hyperref}

\begin{document}

\begin{frontmatter}



\title{On the thermodynamics of pseudo-elastic material models to reproduce the \textsc{Mullins} effect}


\author[TUC]{Christoph Naumann\corref{cor1}}
\ead{christoph.naumann@mb.tu-chemnitz.de}
\author[TUC]{J\"orn Ihlemann}
\ead{joern.ihlemann@mb.tu-chemnitz.de}
\cortext[cor1]{Corresponding author}
\address[TUC]{Technische Universit\"at Chemnitz,
Reichenhainer Stra{\ss}e 70, 09126 Chemnitz}

\begin{abstract}
This work focuses on the thermodynamics of pseudo-elastic models which represent the \textsc{Mullins} effect. Two established models are analyzed theoretically, their thermomechanical properties are derived, and certain critical points are identified. These findings are used to deduce an alternative approach to deviate pseudo-elasticity. This is achieved by defining a suitable free energy which imposes conditions on the stress tensor and the dissipation using the \textsc{Clausius-Duhem} inequality. The concept of pseudo-elasticity is then generalized to extend arbitrary thermomechanical, even inelastic, material models to allow for softening effects. Under weak assumptions on the softening function the thermomechanical consistency is shown.
\end{abstract}

\begin{keyword}
Mullins effect \sep softening \sep pseudo-elasticity \sep thermomechanical consistency
\end{keyword}

\end{frontmatter}

% \linenumbers
%
% \section{ToDo}

% \begin{itemize}
%  \item \sout{Abstract ausarbeiten}
%  \item \sout{Einleitung ausformulieren}
%  \item \sout{Mikromechanisch motivierte Modelle aufschreiben}
%  \item \sout{belastunsginduzierte Anisotropie}
%  \item \sout{Coleman-Noll-paper zitieren}
%  \item \sout{Zitate richtig einfügen}
%  \item \sout{Grundlagenkapitel überarbeiten}
%  \item \sout{$\mathcal D$ durch $\rhot\mathcal D$ ersetzen}
%  \item \sout{Plot von T1 über k fuer OR-Ansatz}
%  \item \sout{T1 über k fuer OR-Ansatz: virgin curve als dotted darstellen}
%  \item \sout{ Beweis Unsymmetrie bei EB-Modell}
%  \item \sout{Beweis zur Polykonvexität durchführen}
%  \item \sout{ Polykonvexität: Ball/Schröder-Neff lesen, warum Existenz etc. ...}
%  \item \sout{ Inelastisches Materialmodell implementieren und Ergebnis zeigen}
%  \item \sout{Anhang: Herleitung, dass die Dissipation-rate von OR das gleiche ist wie die Dissipation, die ich hergeleitet habe}
%  \item \sout{Bildunterschriften ausformulieren}
%  \item \sout{Diagramme: englische Zahlen}
%  \item \sout{partielle Ableitung unter Nebenbedingung noch falsch}
%  \item \sout{Viskoelastizitaet: tau als Materialparameter ersetzen}
%  \item \sout{Viskoelastizitaet: Diskussion der Ergebnisse}
%  \item \sout{Anhang: Beweis der Unsymmetrie der Materialsteifigkeit des ZB-Ansatzes: Text besser schreiben}
%  \item \sout{Zusammenfassung schreiben}
% \end{itemize}

\section{Introduction}
\label{sec:Introduction}The mechanical material behavior of filled rubber is strongly affected by the \textsc{Mullins} effect (\cite{Mullins1948}). Due to this effect, the stress is significantly reduced depending on the maximum load which occurred in the previous history of the material. The micromechanical origin of the \textsc{Mullins} effect has been discussed extensively and yet there is no general agreement in the scientific community. Mullins himself introduced a hard and a soft rubber phase and attributed the softening effect to the conversion of hard into soft rubber by prestraining. \cite{Bueche1960} suggested that the debonding of chains from fillers leads to the softening of the material. The breakdown of filler clusters was proposed by \cite{Kluppel2000} as the reason for the \textsc{Mullins} effect. \cite{Marckmann2002} suggested that the breaking of crosslinks and weak bonds is responsible for the softening. The theory of Self-Organizing Linkage Patterns (SOLP) explains the \textsc{Mullins} effect as a result of a self-organization process based on the weak physical bonds (\cite{Besdo2003a}). Recently it was demonstrated that this concept also provides explanations for other phenomena of rubber behavior (\cite{Wulf2011}). Various other concepts have been suggested, for instance slipping of chain molecules or molecule disentanglement. An extensive overview of physical explanations for the \textsc{Mullins} effect is given in \cite{Diani2009}.\\
Especially in heterogeneously loaded components this effect has to be considered, as different regions with varying stiffness evolve. Due to this resulting inhomogeneous distribution of the material properties, the overall behavior of the component is affected considerably. Therefore, the \textsc{Mullins} effect has to be taken into account to simulate the behavior of elastomeric components.\\
Several micromechanically motivated models have been developed. A model based on the decomposition of the rubber network into a purely elastic polymer part and a polymer-filler network which provides the softening (\cite{Govindjee1991} and \cite{Govindjee1992}) adopts the assumption of \cite{Bueche1960}. The idea of hard and soft rubber phases of \cite{Mullins1948} was used as a basis for a model derived by \cite{Qi2004}.\\
As the micromechanical origin of the \textsc{Mullins} effect is yet not completely understood, phenomenological models still enjoy a great popularity. Especially models based on continuum damage mechanics are widely used. There, a scalar damage variable is introduced which reduces a suitable basic hyperelastic stress (e.g. \cite{Miehe1995a,Chagnon2004}). The main advantage of these models is that the thermomechanical consistency can easily be achieved by defining an appropriate evolution law for the damage variable. Moreover, this approach can be used for extending existing even inelastic models to allow for the Mullins effect (cf. \cite{Simo1987,Miehe2000,Lion1996}). However, the frequently used simple exponential evolution equations are often not capable of reproducing the complex softening effects and more complicated evolution laws have to be defined (\cite{Kaliske2001}). This may lead to a large number of material parameters which complicates a reliable parameter identification significantly.\\
Another approach to represent softening effects are the so-called pseudo-elastic models proposed by \cite{Ogden1999}. These models are based on a hyperelastic description of the virgin material behavior. During unloading and successive reloading the material responds elastically but softer. This is achieved by scaling the basic hyperelastic stress with an appropriate softening function. The main advantage of this family of models is that it is relatively easy to find a suitable softening function as its shape can be directly derived from experimental data (cf. \cite{Kazakeviciute-Makovska2007}). These pseudo-elastic models are widely used and enhanced due to their simplicity (\cite{Ogden2001,Zuniga2002,Dorfmann2003,Dorfmann2004,Guo2006,Pena2009,Zhang2011}).\\
The aforementioned models represent only the isotropic \textsc{Mullins} effect. However, there is experimental evidence that this effect induces anisotropy which cannot be reproduced by these models. However, to represent these anisotropic effects a one-dimensional version for uniaxial tension and compression may be used within directional approaches. It has been shown that this is an effective technique (e.g. \cite{Goktepe2005,Itskov2010,Freund2010,Diercks2013}).\\
The paper has the following structure: Section \ref{sec:pre} deals with the introduction into the used notation and basics of continuum thermodynamics. In section \ref{sec:PE} the basic equations of pseudo-elastic material models are briefly reviewed. Section \ref{sec:OR} shows the approach used in \cite{Ogden1999} which enables the deduction of pseudo-elastic material models by a special free energy function. There, critical points are reviewed. Another model based on the idea of pseudo-elasticity is introduced in \cite{Zuniga2002}. The assumptions made therein are revisited in section \ref{sec:EZB} and critical properties are identified. In section \ref{sec:mine} an alternative approach is presented which enables a consistent deduction of pseudo-elastic models by introducing a suitable free energy. The insertion into the \textsc{Clausius-Duhem} inequality and application of standard arguments directly yields the stress tensor and the dissipation. Thermomechanical consistency is achieved under weak assumptions on the properties of the softening function. Basic properties of the resulting model are deduced and compared to the existing approaches described in the foregoing sections. In the last section, the concept of pseudo-elasticity is extended to allow for arbitrary, even inelastic, models as basic material models. It is shown that the thermomechanical consistency is assured under weak assumptions on the properties of the softening function.

\section{Preliminaries}
\label{sec:pre}
Throughout this article, underlined symbols denote tensors in $\mathbb R^3$ whereby the number of lines determines the rank. In this work, $\I$ stands for the second-rank identity tensor. The deviatoric part of a tensor is defined as $\left( \X \right)' = \X - \frac{1}{3}\left( \X\ppkt\I \right)\I$. The material time derivative is denoted by a superimposed triangle. The superimposed index $G$ denotes the unimodular part of a second rank tensor, i.e. $\X!^G = \det\von\X^{-1/3}\X$. The dyadic product is denoted by $\circ$.\\
One focus of this paper lies on the investigation of thermodynamic properties of material models. Therefore, the second law of thermodynamics is needed:
\begin{equation}
 \rhot \mathcal D_a = \frac{1}{2}\Ttilde\ppkt\Ceck - \rhot \PKTs \psi - \rhot s \PKTs\theta  - \frac{1}{\theta} \Vek Q \pkt \Nabtil \theta \geq 0 \label{eq:CDU_a} \ppp
\end{equation}
This inequality is also known as \textsc{Clausius-Duhem} inequality. In equation \eqref{eq:CDU_a}, $\mathcal D_a$ denotes the overall dissipation, $\Ttilde$ the second \textsc{Piola-Kirchhoff} stress tensor, $\ECK\C$ the material time derivative of the right \textsc{Cauchy-Green} tensor, $\rhot$ the density in the reference configuration, $\psi$ the free energy, $s$ the entropy, $\theta$ the absolute temperature, and $\Vek Q$ the \textsc{Piola-Kirchhoff} heat flux vector. For an exhaustive overview on thermomechanics and a detailed derivation of the inequality \eqref{eq:CDU_a} the reader is referred to \cite{Haupt1999}.\\
Following \cite{Coleman1963a} one possibility to consider the second law of thermodynamics \eqref{eq:CDU_a} is to formulate the material models \textit{a priori} in such a way that no thermomechanical process exists which may violate the \textsc{Clausius-Duhem} inequality. Such material models are then called thermomechanically consistent.\\
If the \textsc{Fourier} model of heat conduction is used for the heat flux $\Vek Q$ (with the heat conductivity $\kappa$), the inequality \eqref{eq:CDU_a} can be divided into two parts:
\begin{alignat}{1}
 &\Vek Q = - \kappa \Nabtil \theta, \quad\mathcal{D}_q := \frac{\kappa}{\theta} \Nabtil \theta \pkt \Nabtil \theta \ge 0,\quad \Rightarrow \mathcal D_a = \mathcal D + \mathcal{D}_q \ge 0  \\
 &\text{for} \quad \mathcal{D} :=  \frac{1}{2}\Ttilde\ppkt\Ceck - \rhot \PKTs \psi - \rhot s \PKTs\theta \ge 0 \ppp \label{eq:CDU}
\end{alignat}
Equation \eqref{eq:CDU} defines the mechanical part of the \textsc{Clausius-Duhem} inequality and the (mechanical) dissipation $\mathcal{D}$. It is assumed that for thermomechanically consistent material models the dissipation $\mathcal{D}$ has to take non-negative values which obviously leads to a non-negative overall dissipation.\\
When dealing with purely mechanical material models the isothermal approximation is often used. There, any change in temperature is neglected, i.e. $\PKTs\theta = 0$. The inequality \eqref{eq:CDU} can then be simplified to
\begin{equation}
 \rhot \mathcal D = \frac{1}{2}\Ttilde\ppkt\Ceck - \rhot \PKTs \psi \ge 0 \ppp \label{eq:CDUisot}
\end{equation}
In the following section \ref{sec:PE} this restriction is used as only isothermal processes are examined by \cite{Ogden1999} and \cite{Zuniga2002}. The generalization to thermomechanical processes is described in section \ref{sec:general}.

\section{Pseudo-elastic material models}
\label{sec:PE}
The basis of pseudo-elastic material models is a hyperelastic basic stress $\Ttilde_0$, which is derived from a suitable free energy function $\psi_0\von\C$:
\begin{equation}
 \Ttilde_0 := 2 \ABL \rhot \psi_0 \nach \C \ppp
\end{equation}
The key idea of pseudo-elastic material models is that the stress during the first loading process is equal to the basic stress $\Ttilde_0$. Upon unloading and reloading the basic stress is scaled by a positive softening function $\eta$. The function $\eta$ thereby depends on a measure $\xi$ for the current load and its maximum $\xi_{max}\von t = \max{\left\{\xi\von\tau, \tau\le t \right\}}$ in the history of the material:
\begin{equation}
 \Ttilde = \eta\von{\xi, \xi_{max}} \Ttilde_0, \quad \text{where } \eta
 \begin{cases}
  = 1, \quad & \xi = \xi_{max},\\
  < 1, & \xi < \xi_{max}
 \end{cases} \ppp \label{eq:PE_Ttilde}
\end{equation}
Different models make use of different measures $\xi$ for the load in the material. In the following two sections, the model of \cite{Ogden1999} and \cite{Zuniga2002} are investigated, which use the free energy of the basic model and a norm of the right \textsc{Cauchy-Green} tensor to quantify the load intensity, respectively. The resulting thermodynamic properties are discussed.

\subsection{The model of Ogden and Roxburgh}
\label{sec:OR}
To derive pseudo-elastic models, \cite{Ogden1999} define the following free energy function:
\begin{equation}
 \psi_{OR}\von{\C, \eta} = \eta \psi_0\von\C + \phi\von\eta \ppp\label{eq:OR_psi}
\end{equation}
In equation \eqref{eq:OR_psi} $\psi_0$ denotes the free energy of the basic material model and $\phi$ is a function which only depends on the softening function $\eta$.
By using equilibrium arguments, the authors demand that the partial derivative of the free energy with respect to the softening function vanishes:
\begin{equation}
 \ABL\psi_{OR}\nach\eta \Sollsein 0 = \psi_0 + \ABL\phi\nach\eta \ppp \label{eq:OR_eta}
\end{equation}
This implicitly defines the softening function $\eta$ depending on $\psi_0$, i.e. the free energy of the basic material model is used as a measure for the current load in the material. Consequently, as $\eta = 1$ for the maximum load, Ogden and Roxburgh use the maximum of the free energy of the basic material model as a measure for the maximum load in the history:
\begin{equation}
\eta = \eta \von{\psi_0, \psi_{max}} , \quad \text{where} \quad \psi_{max}:=\max{\left\{\psi_0\von\tau, \tau\le t \right\}}
\end{equation}
Moreover, equilibrium arguments define the second \textsc{Piola-Kirchhoff} stress as the partial derivative of the free energy with respect to the right \textsc{Cauchy-Green} tensor:
\begin{equation}
 \Ttilde = 2 \ABL\rhot\psi_{OR}\nach\C = 2 \eta \ABL\rhot\psi_0\nach\C = \eta \Ttilde_0 \ppp \label{eq:OR_Ttilde}
\end{equation}
This is the desired formula for the pseudo-elastic material model. The relations \eqref{eq:OR_eta} and \eqref{eq:OR_Ttilde} are sufficient to show the thermodynamic properties of the model. Toward that end, the time derivative of the ansatz \eqref{eq:OR_psi} is inserted into the \textsc{Clausius-Duhem} inequality \eqref{eq:CDUisot}:
\begin{alignat}{1}
 \rhot \mathcal D = \frac{1}{2}\Ttilde\ppkt\ECK\C - \rhot \PKTs\psi_{OR} &= \frac{1}{2}\Ttilde\ppkt\ECK\C - \ABL\rhot\psi_{OR}\nach\C \ppkt \ECK\C - \ABL\rhot\psi_{OR}\nach\eta \PKTs\eta\\
 &=\left(\frac{1}{2}\Ttilde - \eta\ABL\rhot\psi_0 \nach\C\right)\ppkt\ECK\C - \rhot \left( \psi_0 + \ABL\phi\nach\eta \right)\PKTs\eta \label{eq:OR_CDU}
\end{alignat}
Insertion of \eqref{eq:OR_eta} and \eqref{eq:OR_Ttilde} directly leads to the dissipation $\mathcal D = 0$ for this ansatz. Thus, the energy needed to cause the softening in the material is not dissipated and converted into heat but stored in the material. However, all micromechanical theories of the \textsc{Mullins} effect predict that it is at least partly dissipative.
Thus, \cite{Ogden1999} define a ``dissipation rate'' (which represents a physical quantity equal to the dissipation $\mathcal D$) as the rate of the ``residual energy'' which accumulates in the material due to the softening effects. For the two softening functions used in \cite{Ogden1999} and \cite{Dorfmann2004} it can be shown that this dissipation is positive and the authors therefore conclude that the model is consistent with the second law of thermodynamics. However, the positive dissipation does not follow from exploiting the \textsc{Clausius-Duhem} inequality \eqref{eq:CDUisot}. This drawback can be avoided by using the suitable free energy described in section \ref{sec:mine}.\\
To illustrate the behavior of the \textsc{Ogden-Roxburgh} model, a simple shear test with increasing amplitude is performed. For a given orthonormal basis $\left\{\e_x, \e_y, \e_z\right\}$ the following deformation gradient is defined:
\begin{equation}
 \F := \I + k\von t \e_x \circ \e_y \ppp \label{eq:shear}
\end{equation}
The temporal course of the shear deformation $k\von t$ has the following form (figure \ref{fig:kt_OR}):
\begin{equation}
 k\von t := \frac{1}{2}\hat k \left[ 1 - \exp{\left( -\frac{t}{\tau} \right)} \right]\left( 1 - \cos(\omega t) \right) \ppp\label{eq:OR_kt}
\end{equation}
The basic material model is Neo-Hookean. Moreover, the softening function introduced in \cite{Ogden1999} is used:
\begin{alignat}{1}
 \rhot\psi_0 :=& C_{10}\left( \CG \ppkt \I - 3 \right),\label{eq:NH_psi}\\
 \Ttilde_0 =& 2 C_{10} \left( \CG \right)'\pkt \inv\C,\label{eq:NH_Ttilde}\\
 \eta :=& 1 - \frac{1}{r}\operatorname{erf}\left[\frac{1}{m} \left(\psi_{max} - \psi_0  \right) \right] \ppp\label{eq:OR_etaerf}
\end{alignat}
The material parameters are set to $C_{10} = 1 \text{MPa}$, $r = 1$, and $m = 1 \text{MPa}$.
\begin{figure}[h]
	\centering
		\includegraphics[width=0.4\textwidth]{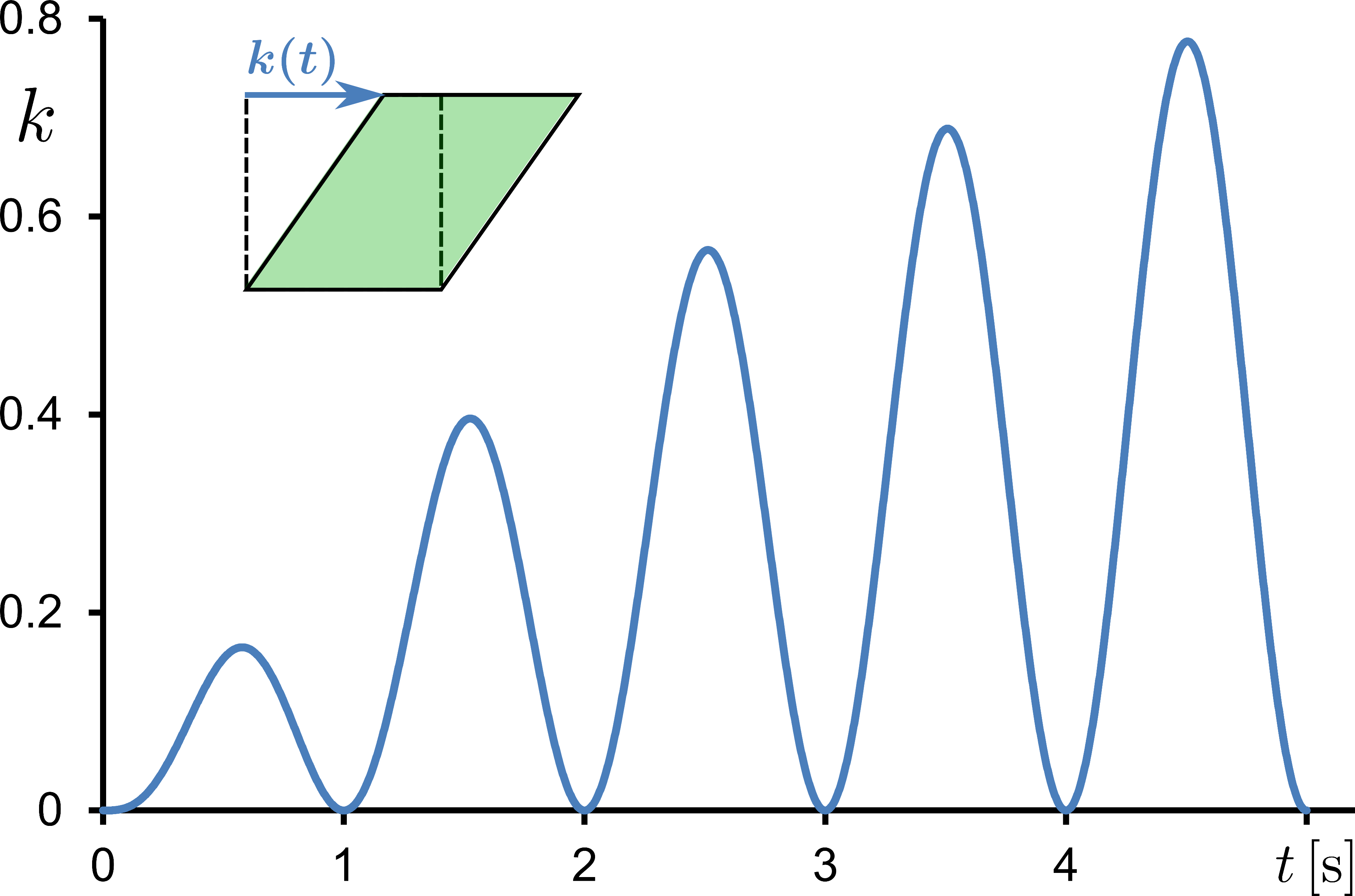}
	\caption{Temporal course of the shear strain $k\von t$ which is used to test the approach of \cite{Ogden1999}.}
	\label{fig:kt_OR}
\end{figure}
In figure \ref{fig:OR_T1_psi} the first \textsc{Piola-Kirchhoff} stress $T_{yx}$ is plotted against the shear strain $k$. The upper dotted curve describes the behavior of the Neo-Hookean basic material model. This stress is computed by the pseudo-elastic material model, if the current deformation is the maximum deformation in the history of the material. When the current deformation is below the maximum deformation the stress is scaled by a positive value $0 \le \eta < 1$. Thus, the pseudo-elastic model predicts a stress response on the lower solid curves for unloading and reloading up to the maximum deformation. The area between the virgin loading path (dotted) and the unloading path (solid) is the dissipated energy due to the softening. However, the special free energy proposed in \cite{Ogden1999} predicts that this energy is not converted into heat but stored in the material which directly follows from equation \eqref{eq:OR_CDU}. This fact is clarified in figure \ref{fig:OR_T1_psi}. There, the free energy $\psi_{OR}$ is plotted against time.
\begin{figure}[h]
	\centering
		\includegraphics[width=1.0\textwidth]{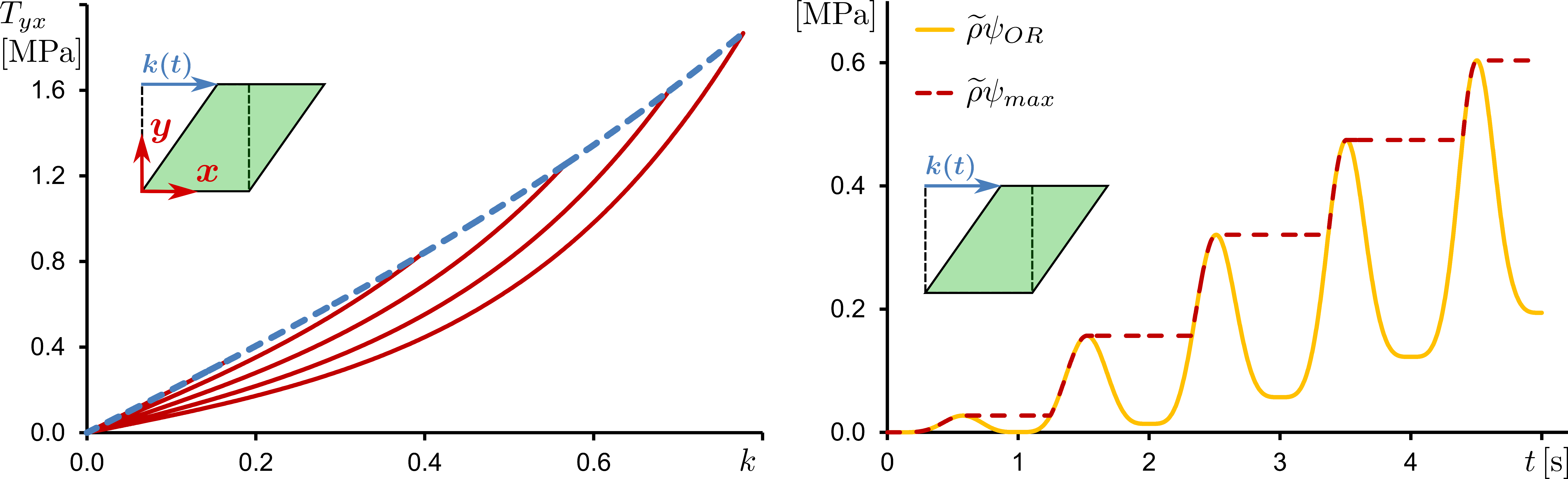}
	\caption{Left: Plot of the pseudo-elastic first \textsc{Piola-Kirchhoff} stress $T_{yx}$ for a simple shear test with increasing amplitude ($k\von t$ from equation \eqref{eq:OR_kt}); Right: Evolution of the free energy $\psi_{OR}$ of the \textsc{Ogden-Roxburgh} model (equation \eqref{eq:OR_psi}) and the history function $\psi_{max}\von t=\max{\left\{\psi_0\von\tau, \tau\le t \right\}}$}
	\label{fig:OR_T1_psi}
\end{figure}
Obviously there is a non-zero value of the free energy $\psi_{OR}$ when the material is fully unloaded, i.e. $k = 0$. This residual value is the energy which was needed to cause the softening in the material.

\subsection{The model of Elias-Zuniga and Beatty}
\label{sec:EZB}
\cite{Zuniga2002} also define a pseudo-elastic model by equation \eqref{eq:PE_Ttilde}. In their work, a norm of the isochoric right \textsc{Cauchy-Green} tensor is used as a measure $m$ for the deformation:
\begin{equation}
 m = \left\| \CG \right\| = \sqrt{\CG\ppkt\CG}, \quad m_{max}\von t = \max\left\{m\von\tau,\tau\le t \right\}\ppp
\end{equation}
In the model of \cite{Zuniga2002} the pseudo-elastic stress is defined by an arbitrary hyperelastic basic stress $\Ttilde_0$ and a softening function $\eta\von{m, m_{max}}$:
\begin{equation}
 \Ttilde := \eta\von{m, m_{max}}\Ttilde_0 \ppp
\end{equation}
The softening function has again the properties $\eta \von{m_{max}, m_{max}} = 1$ and $0 \le \eta \von{m, m_{max}} \le 1$.
No thermodynamical considerations were made in \cite{Zuniga2002}. However, the authors state that ``The unloading and reloading path is elastic so long as the magnitude of the strain does not exceed its greatest previous value''. This statement will be used to reveal an unfavorable property of this model.
Once the maximum of the deformation is fixed, the material response is elastic in the sense that the stress only depends on the current deformation and no dissipation occurs. The behavior, however, is not hyperelastic, i.e. the saved energy in the material depends on the loading path. This leads to an unphysical material response which is shown in the following.\\
Consider a temporal course of the deformation $\C\von t, t \ge t_0$ where $m$ does not exceed a formerly achieved value $m_{max}$ . Then the material behaves elastically, i.e. the dissipation $\mathcal D$ vanishes. Thus, the free energy can be directly determined by integrating the isothermal \textsc{Clausius-Duhem} inequality \eqref{eq:CDUisot} over time:
\begin{equation}
 m_{max} = \text{const.} \quad \Rightarrow \quad \rhot \mathcal{D} = 0 = \frac{1}{2} \Ttilde \ppkt \Ceck - \rhot \PKTs \psi_{EB} \quad \Rightarrow
 \quad \rhot \psi_{EB}\von\tau = \rhot \psi_{EB}\von{t_0} + \int\limits_{t_0}^\tau \frac{1}{2} \Ttilde\ppkt \Ceck \mathrm{d}t \ppp \label{eq:EB_psi}
\end{equation}
The integral which determines the free energy $\psi_{EB}$ is only independent of the path of the deformation $\C\von t$, if the derivative of $\Ttilde$ with respect to the right \textsc{Cauchy-Green} tensor $\C$ is symmetric:
\begin{equation}
 \psi_{EB}\von\tau = \psi_{EB}\von{\C\von\tau, m_{max}} \quad \Leftrightarrow \X \ppkt \ABL \Ttilde \nach \C \ppkt \Y =  \Y \ppkt \ABL \Ttilde \nach \C \ppkt \X  \label{eq:cond_symm} \ppp
\end{equation}
Here, $\X$ and $\Y$ denote arbitrary symmetric second rank tensors. For the model of \cite{Zuniga2002} this is only fulfilled for special cases, where the free energy of the basic model is only a function of $m$. But even with the Neo-Hookean model as a basic model the integrability condition \eqref{eq:cond_symm} is violated. In an analytical manner this is shown in \ref{sec:a_int}. A numerical example is given in the following. The basic material model is Neo-Hookean (equation \eqref{eq:NH_psi} and \eqref{eq:NH_Ttilde}) and the softening function proposed in \cite{Zuniga2002} is applied:
\begin{equation}
 \eta\von{m, m_{max}} = \exp{\left( -b\sqrt{m_{max} - m} \right)} \ppp
\end{equation}
In the following numerical test $C_{10} = 1 \text{MPa}$ and $b = 1$ are used. To demonstrate the drawbacks of this approach, a combined shear-tension test is performed. For a given orthonormal basis $\left\{ \e_x, \e_y, \e_z\right\}$ the following incompressible deformation gradient $\F$ is defined:
\begin{equation}
 \F = \lambda\von t \e_x\circ \e_x + \frac{1}{\sqrt{\lambda\von t}} \left(\e_y\circ \e_y +  \e_z\circ \e_z\right) + k\von t \e_x\circ \e_y \ppp \label{eq:EB_F}
\end{equation}
The temporal course of the stretch $\lambda$ and the shear $k$ are defined as follows:
\begin{alignat}{1}
       k\von t &= \frac{1}{2}\hat{k} \left( \sin(\omega t) + 1 \right),\label{eq:EB_kt}\\
      \lambda\von t &= 1 + \frac{1}{2}\hat\varepsilon\left( 1 - \cos(\omega t) \right) \ppp\label{eq:EB_lambda}
\end{alignat}
In this work, $\hat{k} = 2, \hat\varepsilon = 1$, and $\omega = 2\pi \, \text{s}^{-1}$ (cycle duration equals 1 second) are used. The history variable $m_{max}$ is set to 7, a value which will not be attained by this deformation process. In diagram \ref{fig:psi_EZB} the free energy $\psi_{EB}$ from equation \eqref{eq:EB_psi} is plotted against time. As $\psi_{EB}$ is defined up to an additive constant, $\psi_{EB}\von{t=0} = 0$ is used.
\begin{figure}[h]
	\centering
		\includegraphics[width=0.6\textwidth]{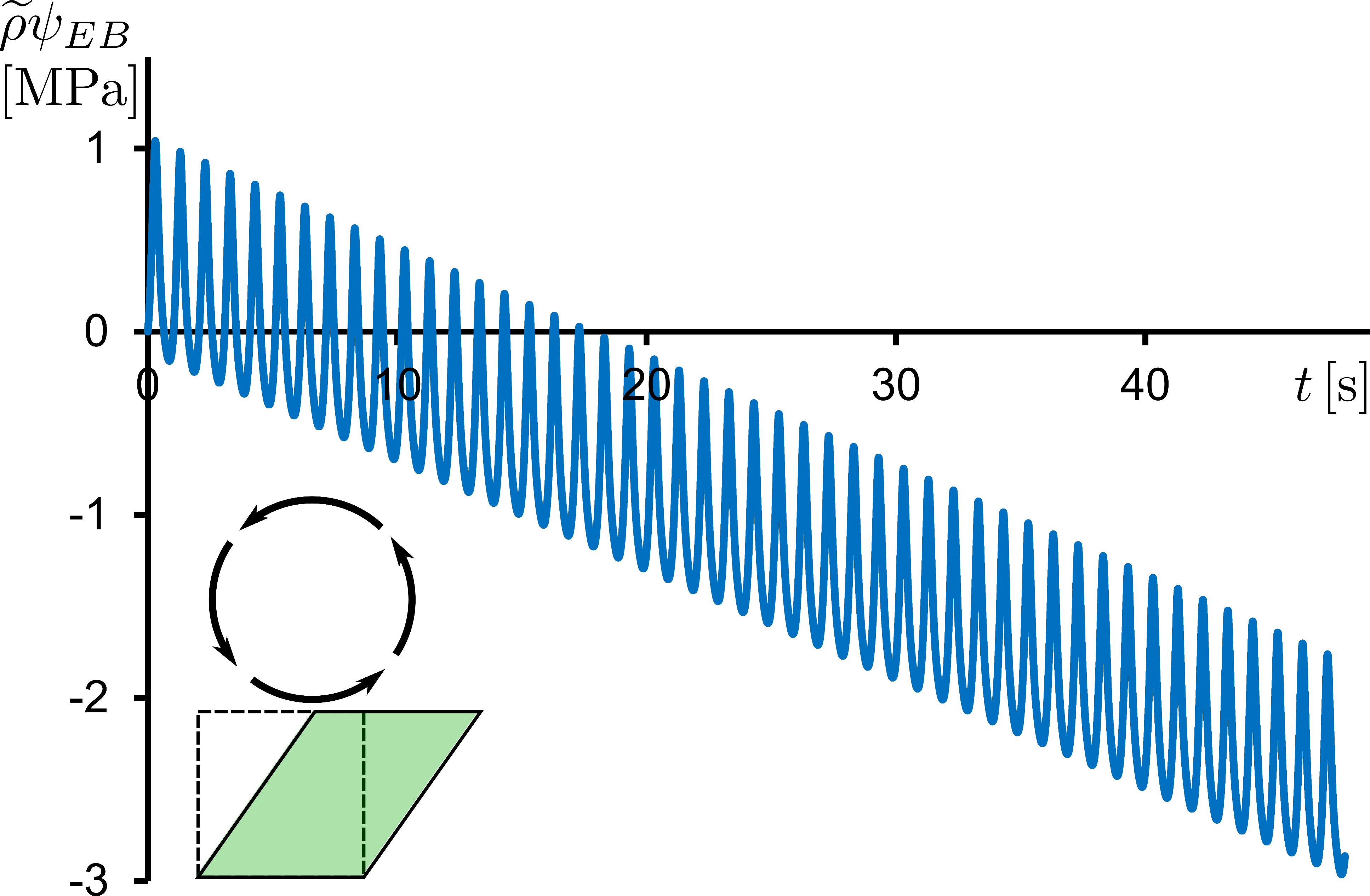}
	\caption{Evolution of the free energy $\psi_{EB}$ of the model of \cite{Zuniga2002} after prestretching for a combined shear-tension test (equations \eqref{eq:EB_F}-\eqref{eq:EB_lambda}).}
	\label{fig:psi_EZB}
\end{figure}
Obviously the energy stored in the material is reduced during every cycle. Thus, this material model predicts that it is possible to reduce the energy in the material boundlessly. In other words, the material provides an infinite supply of energy, which is not observed in real physical problems. Unquestionably, such kind of non-physical models should not be used.

\subsection{Thermomechanical consistency of pseudo-elastic models}
\label{sec:mine}
The aim of this article is to bypass the drawbacks of the approaches described in the previous sections. Thus, the first requirement is that the model behaves hyperelastically during unloading and reloading to rule out an unphysical behavior as described in section \ref{sec:EZB}.\\
This can easily be achieved by following the approach of \cite{Ogden1999} and using the free energy $\psi_0$ of the virgin material model as a measure for the load intensity:
\begin{equation}
 \Ttilde = \eta\von{\psi_0, \psi_{max}} 2 \ABL\rhot\psi_0\nach\C, \quad \text{where }\psi_{max}\von t := \max{\left\{\psi_0\von\tau, \tau\le t \right\}} \label{eq:PE_eta} \ppp
\end{equation}
This will automatically lead to hyperelastic behavior once the maximum load exceeded. To show this, it is sufficient to examine the stiffness-tensor for $\psi_{max} = \text{const.}$. The symmetry of this tensor is a sufficient condition for the integrability of the free energy (cf. equation \eqref{eq:cond_symm}). Differentiating \eqref{eq:PE_eta} we have
\begin{equation}
 \ABL\Ttilde \nach\C = \eta \ABL\Ttilde_0 \nach \C + \Ttilde_0 \circ \ABL\eta\nach\C
  = 2 \eta \ABL^2 \psi_0 \nach{\C^2} + \frac{1}{2} \ABL\eta \nach{\psi_0} \Ttilde_0 \circ \Ttilde_0 \ppp
\end{equation}
This tensor obviously is symmetric and thus hyperelastic behavior is obtained during unloading and successive reloading.\\
The second goal is to derive the constitutive equations directly by evaluating the \textsc{Clausius-Duhem} inequality without using any assumptions deduced from equilibrium-thermodynamics (like equations \eqref{eq:OR_eta} and \eqref{eq:OR_Ttilde}).
As the free energy $\psi_0\von\C$ is defined up to an additive constant it is assumed that in the undeformed configuration $\psi_0\von{\C = \I} = 0$.
The key idea is to define the following free energy for the pseudo-elastic material model:
\begin{equation}
 \psi\von{\psi_0, \psi_{max}} := \int\limits_0^{\psi_0} \eta\von{\xi, \psi_{max}} \mathrm{d}\xi \label{eq:PE_psi} \ppp
\end{equation}
This definition directly leads to $\psi\von{\psi_0 = 0, \psi_{max}} = 0$, i.e. the free energy will be always zero when the material is in the undeformed configuration. Moreover, differentiation yields
\begin{equation}
 \ABL\psi\von{\psi_0, \psi_{max}}\nach{\psi_0} = \eta\von{\psi_0, \psi_{max}} \ppp \label{eq:dpsidpsi0}
\end{equation}
Thus, the time derivative of $\psi$ can be obtained by using the chain-rule and equation \eqref{eq:dpsidpsi0}:
\begin{alignat}{1}
 \rhot\PKTs\psi &= \ABL\rhot\psi\von{\psi_0, \psi_{max}}\nach{\psi_0}\PKTs\psi_0 + \ABL\rhot\psi\von{\psi_0, \psi_{max}}\nach{\psi_{max}}\PKTs\psi_{max}\\
 &= \eta\ABL\rhot\psi_0\von\C\nach\C\ppkt\ECK\C + \ABL\rhot\psi\von{\psi_0, \psi_{max}}\nach{\psi_{max}}\PKTs\psi_{max}\\
 &= \eta \frac{1}{2} \Ttilde_0\ppkt\ECK\C + \ABL\rhot\psi\von{\psi_0, \psi_{max}}\nach{\psi_{max}}\PKTs\psi_{max} \label{eq:PE_psipkt} \ppp
\end{alignat}
Inserting \eqref{eq:PE_psipkt} into the \textsc{Clausius-Duhem} inequality \eqref{eq:CDUisot} we obtain
\begin{equation}
 \rhot \mathcal D = \left(\frac{1}{2}\Ttilde - \frac{1}{2} \eta \Ttilde_0 \right)\ppkt \ECK\C
  - \ABL\rhot\psi\nach{\psi_{max}}\PKTs\psi_{max} \ge 0 \label{eq:PE_CDU} \ppp
\end{equation}
This inequality is fulfilled for arbitrary $\ECK\C$, if the following conditions hold:
\begin{alignat}{1}
 \Ttilde &= \eta \Ttilde_0 \label{eq:PE_ttildecdu} \\
 \ABL\rhot\psi\nach{\psi_{max}}\PKTs\psi_{max} &\le 0 \label{eq:PE_D} \ppp
\end{alignat}
Condition \eqref{eq:PE_ttildecdu} is the desired stress formula of pseudo-elastic models which is now directly derived from the \textsc{Clausius-Duhem} inequality by using standard arguments.
If additionally the inequality \eqref{eq:PE_D} is satisfied, the dissipation is always positive and thermomechanical consistency is assured. In the following, we assume that $\eta\von{\psi_0, \psi_{max}}$ is continuous and continuously differentiable with respect to $\psi_{max}$. Then, integration and differentiation can be interchanged and the left hand side of inequality \eqref{eq:PE_D} can be written as:
\begin{equation}
 \ABL\psi\nach{\psi_{max}}\PKTs\psi_{max} = \ABL\nach{\psi_{max}}\int\limits_0^{\psi_0} \eta\von{\xi, \psi_{max}} \mathrm{d}\xi \, \PKTs\psi_{max}
 = \int\limits_0^{\psi_0} \ABL\eta\von{\xi, \psi_{max}}\nach{\psi_{max}} \mathrm{d}\xi \, \PKTs\psi_{max} \ppp
\end{equation}
Due to the definition of the history variable (equation \eqref{eq:PE_eta}), $\PKTs\psi_{max}$ is always non-negative. Thus, a sufficient condition for thermomechanical consistency is
\begin{equation}
 \ABL\eta\von{\psi_0, \psi_{max}}\nach{\psi_{max}} \le 0 \ppp
\end{equation}
This constraint can be easily interpreted: every softening function $\eta\von{\psi_0, \psi_{max}}$ which is a monotonically decreasing function of $\psi_{max}$ for every fixed $\psi_0$ leads to a positive dissipation. Thus, every function $\eta\von{\psi_0, \psi_{max}}$ which leads to further softening for larger deformations is admissible. For the softening functions proposed in \cite{Ogden1999} and \cite{Dorfmann2004} the appropriate free energy and dissipation are given in \ref{sec:a_psiOR}.\\
To compare the free energy $\psi_{OR}$ of \cite{Ogden1999} from equation \eqref{eq:OR_psi} and the new proposed $\psi$ from equation \eqref{eq:PE_psi} a plot of these free energies for the shear test described in section \ref{sec:OR} is depicted in figure \ref{fig:compare_psiOR_PE}. The basic material model is Neo-Hookean and the softening function from equation \eqref{eq:OR_eta} is used. The same material parameters and shear deformation cycle as in section \ref{sec:OR} are used.
\begin{figure}[h]
	\centering
		\includegraphics[width=0.7\textwidth]{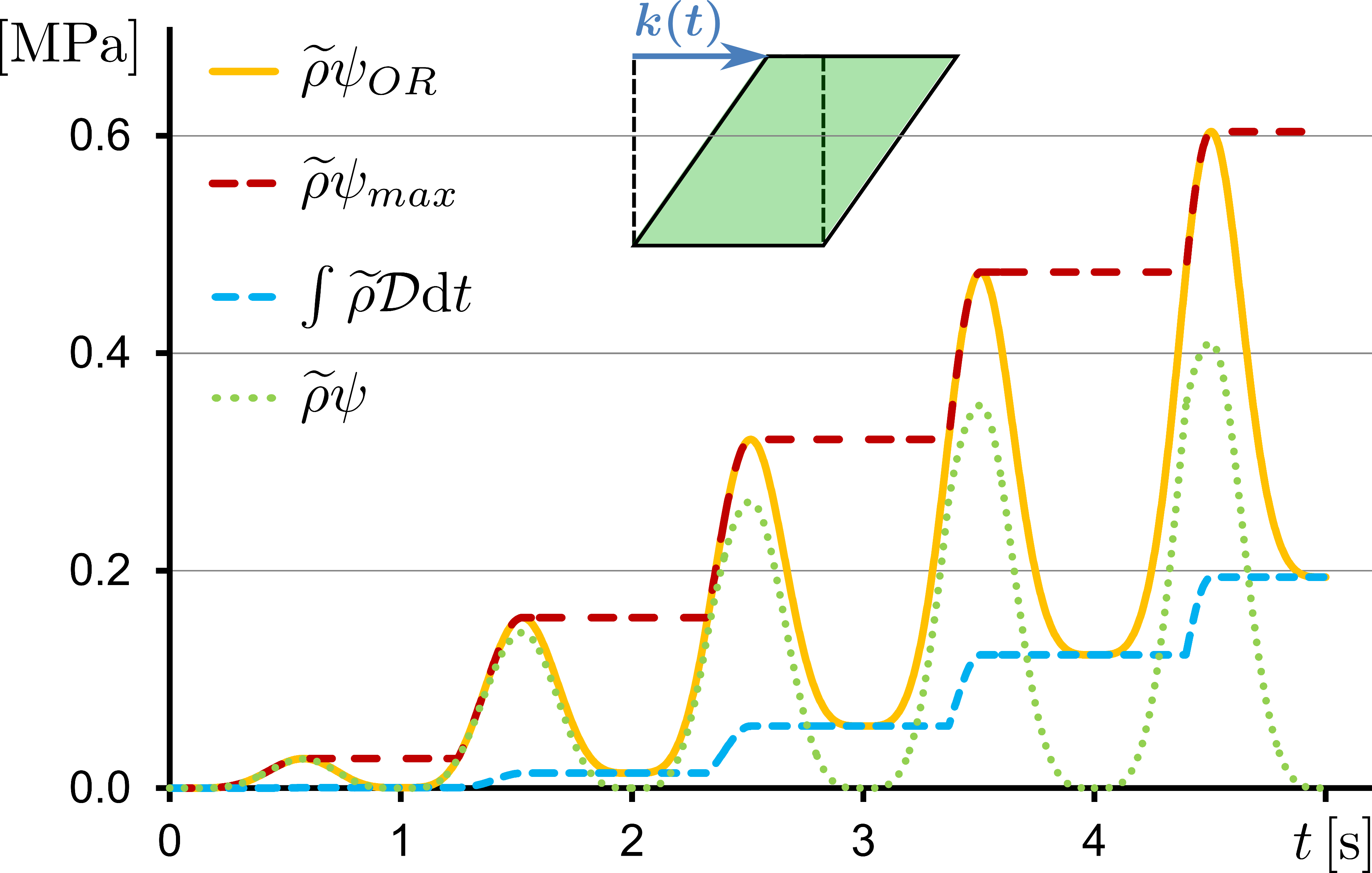}
	\caption{Evolution of the free energy $\psi_{OR}$ of the \textsc{Ogden-Roxburgh} model (equation \eqref{eq:OR_psi}), the history function $\psi_{max}\von t=\max{\left\{\psi_0\von\tau, \tau\le t \right\}}$, and the newly proposed free energy $\psi$ from \eqref{eq:PE_psi} for a simple shear test. The difference between the free energies $\psi_{OR}$ and $\psi$ is the dissipated energy $\int \mathcal D \mathrm d t$.}
	\label{fig:compare_psiOR_PE}
\end{figure}
It can be clearly seen that the proposed free energy \eqref{eq:PE_psi} always takes the value $0$ if the material is in the undeformed configuration and no residual energy accumulates in the material. Moreover, the difference between the free energies $\psi_{OR}$ and $\psi$ is the dissipated energy $\int \mathcal D \mathrm d t$ which is predicted by the left-hand side of \eqref{eq:PE_D}.

\begin{rmk}
As soon as the maximum load is reached, the behavior of the softening model is locally hyperelastic. Then, it may be desirable that the corresponding free energy is polyconvex to assure the existence of minimizers of the related variational principle (\cite{Ball1976}). For a polyconvex free energy of the basic material model this property will be preserved by the softening model given that $\eta\von{\psi_0, \psi_{max}}\ge 0$ and $\ABL\eta\nach{\psi_0} \ge 0$. Then, it directly follows from the definition \eqref{eq:PE_psi} of the free energy that $\psi\von{\psi_0, \psi_{max}}$ is a monotonic and convex function of $\psi_0$ for every fixed $\psi_{max}$. This directly yields the polyconvexity of the free energy $\psi$ (cf. \cite{Schroder2003}).
\end{rmk}

\begin{rmk}
The \textsc{Mullins} effect may only be partly dissipative, i.e. a fraction of the energy needed to cause the softening is not dissipated but stored in the material \footnote{ The consideration of such stored energy component becomes important if the equation of heat conduction is derived directly from the energy balance (\cite{Shutov2011}).}. This additional energy storage may also be considered within this approach. Indeed, the free energy $\psi$ can be decomposed in the following way:
\begin{equation}
 \psi = \int\limits_0^{\psi_0} \eta\von{\xi, \psi_{max}} \mathrm{d}\xi + \psi_s \ppp
\end{equation}
Here, the part $\psi_s$ represents the energy which is stored in the material and not dissipated due to the softening. It can easily be seen that the model is thermomechanically consistent if the following evolution equation for $\psi_s$ is defined:
\begin{equation}
 \PKTs\psi_s = -\gamma\von{\psi_{max}} \int\limits_0^{\psi_{0}} \ABL\eta\von{\xi, \psi_{max}}\nach{\psi_{max}} \mathrm{d}\xi \, \PKTs\psi_{max}, \quad \text{where } 0\leq \gamma\von{\psi_{max}} \leq 1 \label{eq:PE_psis} \ppp
\end{equation}
The material function $\gamma$ may or not depend on $\psi_{max}$ and has to take values between $0$ and $1$ to assure a positive dissipation.
Noting that $\PKTs\psi_{max} \neq 0$ only for $\psi_0 = \psi_{max}$, the evolution equation \eqref{eq:PE_psis} can be rewritten as follows:
\begin{equation}
\PKTs\psi_s = -\gamma\von{\psi_{max}} \int\limits_0^{\psi_{max}} \ABL\eta\von{\xi, \psi_{max}}\nach{\psi_{max}} \mathrm{d}\xi \, \PKTs\psi_{max} \ppp\label{eq:PE_psis1}
\end{equation}
Hence, the evolution equation \eqref{eq:PE_psis1} can be integrated exactly:
\begin{equation}
 \psi_s\von{\psi_{max}} = -\int\limits_0^{\psi_{max}} \gamma\von\zeta \int\limits_0^{\zeta} \ABL\eta\von{\xi, \zeta}\nach{\zeta} \mathrm{d}\xi \, \mathrm d \zeta \ppp
\end{equation}
If the value $\gamma = 1$ is used, the model of \cite{Ogden1999} is obtained; for $\gamma = 0$ the \textsc{Mullins} effect is fully dissipative.
\end{rmk}

\section{Generalization for arbitrary basic material models}
\label{sec:general}
The idea of the pseudo-elastic material models can be generalized to allow for arbitrary thermomechanically consistent material models. Let $\left\{\psi_0, \Ttilde_0, s_0, \mathcal{D}_0 \right\}$ a material model, which satisfies the \textsc{Clausius-Duhem} inequality for arbitrary thermomechanical processes:
\begin{equation}
 \rhot \mathcal{D}_0 = \frac{1}{2}\Ttilde_0\ppkt \Ceck - \rhot\left( \PKTs\psi_0 + s_0 \PKTs\theta \right) \ge 0 \label{eq:CDU_base} \ppp
\end{equation}
No assumptions concerning the structure of the model have to be made, i.e. the material model may for example be defined by internal variables (see \cite{Coleman1967a}) or have a functional form (like defined in \cite{Coleman1964a}).
Under this assumption, the softening effect can be added in a thermomechanically consistent manner as follows:
\begin{thm}
\label{thm:ge}
 Let $\left\{\psi_0, \Ttilde_0, s_0, \mathcal{D}_0 \right\}$ a thermomechanically consistent material model, $\psi_{max}\von t := \max{\left\{\psi_0\von\tau, \tau\le t \right\}}$, and $\eta\von{\psi_0, \psi_{max}} \ge 0$ a sufficiently smooth softening function such that $\displaystyle{\ABL\eta\nach{\psi_{max}}} \le 0$. Then the following softening model is thermomechanically consistent:
\begin{alignat}{1}
  \psi\von{\psi_0, \psi_{max}} &= \int\limits_0^{\psi_0} \eta\von{\xi, \psi_{max}} \mathrm{d}\xi \label{eq:ge_psi},\\
 \Ttilde &= \eta\Ttilde_0\label{eq:ge_Ttilde},\\
 s &= \eta s_0 \label{eq:ge_s},\\
 \mathcal D &= \eta \mathcal{D}_0 - \ABL \psi \nach {\psi_{max}} \PKTs\psi_{max} \ge 0 .\label{eq:ge_D} \ppp
\end{alignat}
\end{thm}
\begin{proof}
The time derivative of the free energy of the softening model is
 \begin{equation}
\PKTs\psi = \ABL \psi\nach{\psi_0} \PKTs\psi_0 + \ABL\psi\nach{\psi_{max}}\PKTs{\psi_{max}}
 = \eta \PKTs\psi_0 + \ABL\psi\nach{\psi_{max}}\PKTs{\psi_{max}}\ppp \label{eq:CDU_univ}
\end{equation}
Rearranging the definition \eqref{eq:CDU_base} of the dissipation for the basic material model, we obtain the time derivative of the free energy $\psi_0$ of the basic material model:
\begin{equation}
\rhot\PKTs\psi_0 = \frac{1}{2}\Ttilde_0 \ppkt \Ceck  - \rhot s_0\PKTs\theta  - \rhot\mathcal D_0\ppp \label{eq:CDU_base_psi0}
\end{equation}
Inserting \eqref{eq:CDU_base_psi0} and \eqref{eq:CDU_univ} into the definition of the dissipation (cf. \eqref{eq:CDU}) and rearranging terms yields
\begin{alignat}{1}
\rhot \mathcal D = &\left(\frac{1}{2}\Ttilde - \frac{1}{2}\eta\Ttilde_0\right)\ppkt \Ceck
- \rhot\left(s - \eta s_0\right)\PKTs\theta\label{eq:Univ_Ttilde}\\
&+ \eta \rhot \mathcal D_0 - \rhot \ABL\psi\nach{\psi_{max}}\PKTs{\psi_{max}}\label{eq:Univ_Diss}\ppp
\end{alignat}
The first two terms on the right-hand side of \eqref{eq:Univ_Ttilde} vanish due to \eqref{eq:ge_Ttilde} and \eqref{eq:ge_s}. Thus, to prove thermomechanical consistency, the only requirement is that remaining expression in line \eqref{eq:Univ_Diss} is non-negative. Due to the assumptions, the softening function $\eta$ and the dissipation $\mathcal D_0$ of the basic material model are always non-negative. Thus,
\begin{equation}
\mathcal D = \eta\mathcal D_0 - \ABL\psi\nach{\psi_{max}}\PKTs{\psi_{max}} \ge  - \ABL\psi\nach{\psi_{max}}\PKTs{\psi_{max}} \ppp
\end{equation}
From the definition of $\psi_{max}$ the time derivative of the history variable $\PKTs\psi_{max}$ is always non-negative. Thus, it is sufficient to show that the derivative of the free energy with respect to $\psi_{max}$ is negative. As the softening function is sufficiently smooth, integration and differentiation can be interchanged. From the monotonicity requirement of the softening function and the monotonicity property of the integral it follows that
\begin{equation}
\ABL\psi\nach{\psi_{max}} = \ABL\nach{\psi_{max}}\int\limits_0^{\psi_0}{\eta\von{\xi, \psi_{max}}}\mathrm d \xi
= \int\limits_0^{\psi_0}{\ABL\eta\von{\xi, \psi_{max}}\nach{\psi_{max}}}\mathrm d \xi \le 0 \ppp
\end{equation}
Thus, the overall softening model provides a non-negative dissipation and is therefore thermomechanically consistent.
\end{proof}

\subsection{Example: Thermohyperelasticity}
Using Theorem \ref{thm:ge}, a thermohyperelastic material model may be used as a basic material model. Thermohyperelasticity is deduced from a suitable free energy function $\psi_0\von{\C, \theta}$ (cf. \cite{Haupt1999}):
\begin{equation}
 \Ttilde_0 = 2 \ABL {\rhot\psi_0\von{\C, \theta}}\nach\C, \,\, s_0 = -\ABL\psi_0\von{\C, \theta}\nach\theta, \,\, \mathcal D = 0 \ppp
\end{equation}
Using the free energy \eqref{eq:ge_psi}, the thermomechanical generalization of the pseudo-elastic material models is as follows:
\begin{equation}
 \Ttilde = 2 \eta \ABL\rhot\psi_0\nach\C, \,\, s = -\eta \ABL\psi_0\nach\theta, \,\, \mathcal{D} = \int\limits_0^{\psi_{0}} \ABL\eta\von{\xi, \psi_{max}}\nach{\psi_{max}} \mathrm{d}\xi \, \PKTs\psi_{max} \label{eq:ge_thermoelas} \ppp
\end{equation}
The first equation in \eqref{eq:ge_thermoelas} is the known stress formula of pseudo-elastic models. The second relation is the conditional equation for the entropy $s$, which has to be specified when thermomechanical processes are considered. The dissipation in the third equation is always non-negative under the given restrictions on the softening function $\eta$ (see Theorem \ref{thm:ge}).

\subsection{Example: Viscoelasticity}
To show that even inelastic material models are admissible as a basic material model, the following (isothermal) viscoelastic model of a \textsc{Maxwell} body is investigated:
\begin{alignat}{1}
 \ECK\Ttilde_0 &= G \Ceck - \frac{1}{\tau_0} \Ttilde_0, \quad  \Ttilde_0\von{t = 0} = \Tenz 0,\label{eq:VE_Ttilde0eck}\\
 \rhot\psi_0 &= \frac{1}{4 G} \Ttilde_0 \ppkt \Ttilde_0, \quad \rhot\mathcal{D} = \frac{1}{2 G \tau_0}\Ttilde_0 \ppkt \Ttilde_0 \label{eq:VE_psi0}\ppp
\end{alignat}
Here, the modulus $G$ and the relaxation time $\tau_0$ are material parameters. It can easily be seen that this model is thermomechanically consistent by inserting the time derivative of $\psi_0$ and the evolution equation \eqref{eq:VE_Ttilde0eck} into the \textsc{Clausius-Duhem} inequality \eqref{eq:CDUisot}.\footnote{Using the terminology of \cite{Shutov2014}, this model is not weakly invariant under the reference change.}\\
The corresponding softening model is defined by equations \eqref{eq:OR_etaerf}, \eqref{eq:ge_psi}, \eqref{eq:ge_Ttilde}, and \eqref{eq:ge_D}:
\begin{alignat}{1}
 \eta &= 1 - \frac{1}{r}\operatorname{erf}\left[\frac{1}{m} \left(\psi_{max} - \psi_0  \right) \right]\label{eq:VE_eta},\\
 \psi\von{\psi_0, \psi_{max}} &= \int\limits_0^{\psi_0} \eta\von{\xi, \psi_{max}} \mathrm{d}\xi \label{eq:VE_psi},\\
 \Ttilde &= \eta\Ttilde_0,\label{eq:VE_Ttilde}\\
 \rhot\mathcal{D} &= \frac{\eta}{2 G \tau_0}\Ttilde_0 \ppkt \Ttilde_0 - \int\limits_0^{\psi_0}{\ABL\eta\von{\xi, \psi_{max}}\nach{\psi_{max}}}\mathrm d \xi \PKTs\psi_{max}\label{eq:VE_D}\ppp
\end{alignat}
To visualize the behavior of the resulting model, a two-sided shear test with increasing shear strain is performed. The corresponding deformation gradient is given by \eqref{eq:shear}. The following temporal course of the shear $k$ is used (see also figure \ref{fig:VE_kt}):
\begin{equation}
 k\von t = \hat k \left( 1 - \exp\left[-\frac{t}{\tau} \right]\right) \sin\left(\omega t \right) \label{eq:VE_kt} \ppp
\end{equation}
\begin{figure}[h]
	\centering
		\includegraphics[width=0.4\textwidth]{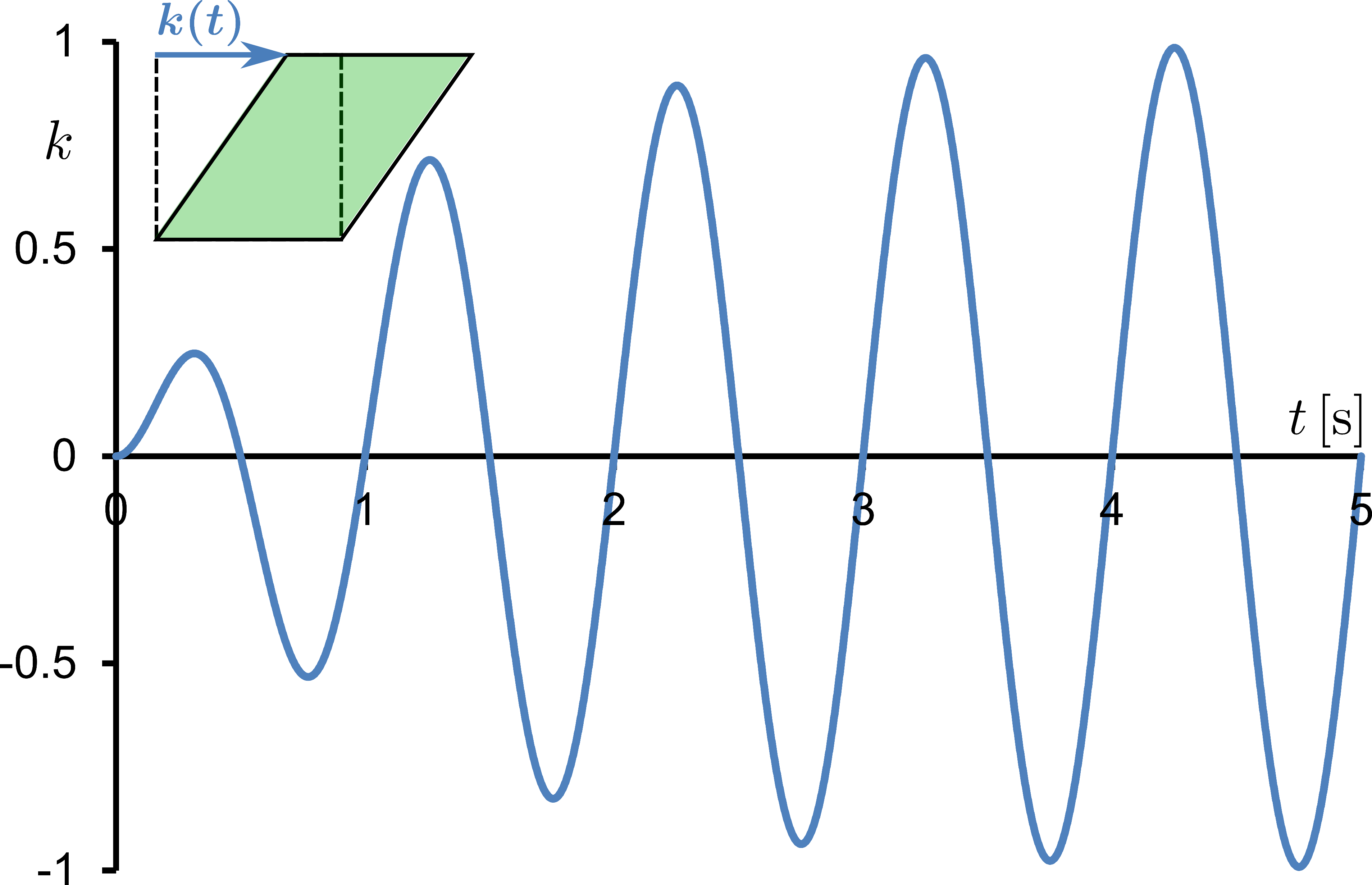}
	\caption{Temporal course of the shear strain $k\von t$ which is used to test the behavior of the viscoelastic softening model.}
	\label{fig:VE_kt}
\end{figure}
The parameters $\hat k = 1$, $\tau = 1 \text{s}$, and $\omega = 2 \pi \text{s}^{-1}$ are used for the shear cycle. Moreover, the material parameters are set to $G = 1 \text{MPa}$, $\tau_0 = 1\text{s}$, $r = 1$, and $m = 1 \text{MPa}$. The resulting behavior of the softening model is plotted in figure \ref{fig:VE_T1}.
\begin{figure}[h]
	\centering
		\includegraphics[width=1.0\textwidth]{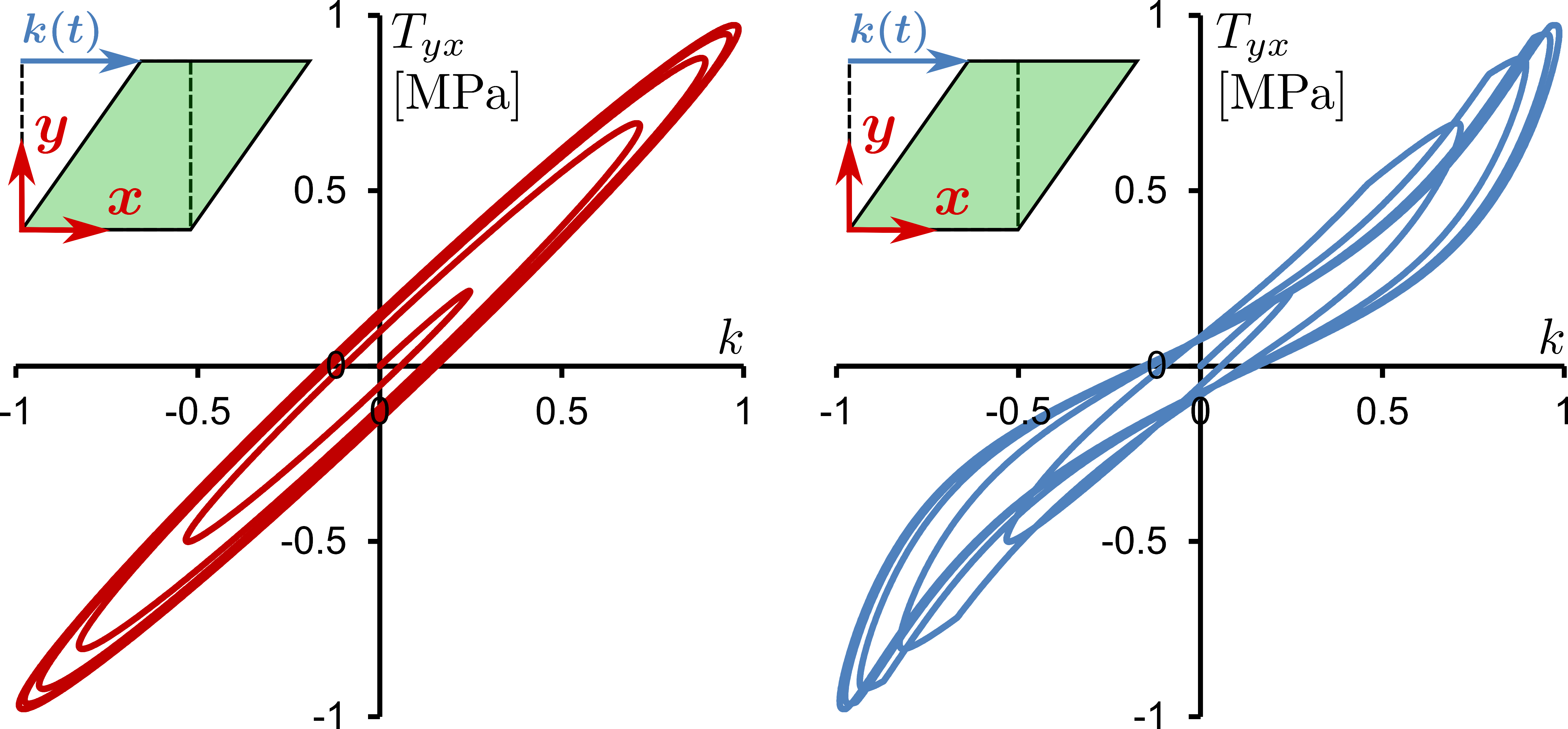}
	\caption{Plot of the first \textsc{Piola-Kirchhoff} stress $T_{yx}$ for a two-sided shear test with increasing amplitude ($k\von t$ from equation \eqref{eq:VE_kt}); Left: Stress-strain behavior of the viscoelastic basic material model; Right: Stress-strain curve of the corresponding softening model.}
	\label{fig:VE_T1}
\end{figure}
There, the behavior of the basic material model (equations \eqref{eq:VE_Ttilde0eck} and \eqref{eq:VE_psi0}) is compared to the softening model (equations \eqref{eq:VE_eta} - \eqref{eq:VE_D}). Obviously, the softening model is able to represent the softening effect depending on the maximum load which occurred in the history of the material. For points where the former maximum stress (and thus the free energy) in the history of the material is exceeded the stress of the softening model is equal to the stress of the basic viscoelastic model. If the current stress is below the maximum stress in the history, the softening function $\eta$ becomes active and a pronounced softening is observed.

\section{Conclusion}
In this paper, thermodynamical properties of pseudo-elastic models to represent the \textsc{Mullins} effect are investigated. The analysis of two established models of pseudo-elasticity provides an insight into critical points of currently used approaches. To overcome these difficulties, an alternative approach is proposed which enables the deviation of new pseudo-elastic models. Toward that end, a suitable free energy is specified, which yields conditional equations for the stress tensor and the dissipation after exploiting the \textsc{Clausius-Duhem} inequality. To bypass the limitation that only hyperelastic models are admissible as basic material models the concept of pseudo-elasticity is generalized. The proposed approach enables the extension of arbitrary thermomechanical, inelastic material models to allow for softening effects. Under natural assumptions on the softening function, the thermomechanical consistency is shown. Different examples demonstrate the applicability of this approach.

\appendix

\section{Integrability}
\label{sec:a_int}
In the following appendix, the path dependence of the free energy $\psi_{EB}$ from equation \eqref{eq:EB_psi} is analyzed.
A necessary condition that the integral in \eqref{eq:EB_psi} is independent of the deformation path is that the derivative of the second \textsc{Piola-Kirchhoff} stress tensor with respect to the right \textsc{Cauchy-Green} tensor is symmetric, i.e. for all symmetric tensors $\X$ and $\Y$ of second rank the relation
\begin{equation}
 \X \ppkt \ABL \Ttilde \nach \C \ppkt \Y =  \Y \ppkt \ABL \Ttilde \nach \C \ppkt \X \label{eq:cond_symm_app}
\end{equation}
is fulfilled. In case of the model of \cite{Zuniga2002} it can be shown that the stiffness is in general not symmetric, i.e. the free energy depends on the path of the deformation. In this section this is shown for the Neo-Hookean model as a basic material model. There, the basic second \textsc{Piola-Kirchhoff} stress is given as
\begin{equation}
 \Ttilde_0 = 2 C_{10} \left( \CG \right)'\pkt \inv\C \ppp
\end{equation}
Using the derivative of the norm $m$ of the isochoric right \textsc{Cauchy-Green} tensor
\begin{equation}
 \ABL m \nach\C = I_3^{-\frac{1}{3}} m \left(\CG - \frac{m^2}{3}\inv\CG  \right)
\end{equation}
and the chain rule yields the following stiffness:
\begin{equation}
 \ABL\Ttilde\nach\C = \eta \ABL\Ttilde_0\nach\C + \Ttilde_0 \circ \ABL \eta\nach m \ABL m \nach \C = \eta  \ABL\Ttilde_0\nach\C + \ABL\eta\nach m \frac{2 C_{10}}{m} I_3^{-\frac{2}{3}}\left( \left( \CG \right)'\pkt \inv\CG \right) \circ \left(\CG - \frac{m^2}{3}\inv\CG  \right) \ppp
\end{equation}
The first part of this stiffness is obviously symmetric as $\Ttilde_0$ is derived from the free energy function $\psi_0$. In contrast, the second part is unsymmetric. To show this it is sufficient to examine the important part $\Tenv M$:
\begin{equation}
 \Tenv M := \left( \left( \CG \right)'\pkt \inv\C \right) \circ \left(\CG - \frac{m^2}{3}\inv\CG  \right) \ppp
\end{equation}

For the two special tensors $\X = \inv\CG$ and $\Y = \I$ the relation \eqref{eq:cond_symm_app} is violated. After some basic algebra it can be shown that
\begin{equation}
 \alpha := \inv\CG \ppkt \Tenv M \ppkt \I = \left[ I_2 - \frac{I_1}{3}\left(I_2^2 - 2 I_1 \right) \right] \left[ I_1 - \frac{I_2}{3}\left(I_1^2 - 2 I_2 \right) \right] \label{eq:Int_alpha}
\end{equation}
and
\begin{equation}
 \beta = \I \ppkt \Tenv M \ppkt \inv\CG = \left[ 3 - \frac{I_1}{3}I_2\right] \left[ 3 - \frac{1}{3}\left(I_1^2 - 2 I_2 \right)\left(I_2^2 - 2 I_1 \right) \right] \ppp \label{eq:Int_beta}
\end{equation}
In equations \eqref{eq:Int_alpha} and \eqref{eq:Int_beta} the following abbreviations were used:
\begin{equation}
 I_1 := I_1\von\CG = \CG \ppkt \I, \,\,\, I_2 := I_2\von\CG = \frac{1}{2}\left( I_1^2 - \CG \ppkt \CG \right) \ppp
\end{equation}
If the stiffness tensor was symmetric then $\alpha - \beta = 0$. In contrast,
\begin{equation}
 \alpha - \beta = \frac{1}{3}\left[-4I_1^3 - 4 I_2^3 + I_1^2I_2^2 + 18 I_1 I_2 - 27\right] \ppp
\end{equation}
In general, this expression is not equal to zero.

\section{Free energies for proposed softening functions}
\label{sec:a_psiOR}
In the following appendix it is shown that the softening functions proposed in \cite{Ogden1999} and \cite{Dorfmann2004} lead to a thermomechanically consistent material model.\\
For the softening function $\eta$ used in \cite{Ogden1999}
\begin{equation}
 \eta = 1 - \frac{1}{r}\operatorname{erf}{\left( \frac{\psi_{max} - \psi_0}{m} \right)}
\end{equation}
the related free energy $\psi$ of the pseudo-elastic model can be computed by integration of $\eta$ (equation \eqref{eq:PE_psi}):
\begin{alignat}{1}
 \psi =& \frac{m}{\sqrt\pi r} \exp{\left( -\frac{\left[\psi_{max} - \psi_0\right]^2}{m^2} \right)}
 + \frac{\psi_{max} - \psi_0}{r} \operatorname{erf}{\left(\frac{\psi_{max} - \psi_0}{m} \right)} + \psi_0 \\
&- \frac{m}{\sqrt\pi r} \exp{\left( -\frac{\psi_{max}^2}{m^2} \right)} - \frac{\psi_{max}}{r} \operatorname{erf}{\left(\frac{\psi_{max}}{m} \right)} \ppp
\end{alignat}
By differentiation, the thermomechanical consistency can be shown:
\begin{equation}
 \mathcal D = -\ABL \psi\von{\psi_0, \psi_{max}}\nach{\psi_{max}}\PKTs\psi_{max} = \frac{1}{r}\operatorname{erf}{\left(\frac{\psi_{max}}{m} \right)}\PKTs\psi_{max} \ge 0 \quad \text{for } r>0, m>0 \ppp\label{eq:A_D_eta1}
\end{equation}
This expression for $\mathcal{D}$ is equal to the quantity, \cite{Ogden1999} denote as ``dissipation rate''.\\
The free energy related to the softening function
\begin{equation}
 \eta =  1 - \frac{1}{r}\tanh{\left( \frac{\psi_{max} - \psi_0}{m} \right)}
\end{equation}
used in \cite{Dorfmann2004} can be given by
\begin{equation}
 \psi = \frac{m}{r}\ln{\left( \cosh{\left[ \frac{\psi_{max} - \psi_0}{m}\right]} \right)} + \psi_0 - \frac{m}{r}\ln{\left( \cosh{\left[ \frac{\psi_{max}}{m}\right]} \right)}  \ppp
\end{equation}
The thermomechanical consistency again follows from differentiation:
\begin{equation}
 \mathcal D = -\ABL \psi\von{\psi_0, \psi_{max}}\nach{\psi_{max}}\PKTs\psi_{max} = \frac{1}{r}\tanh{\left(\frac{\psi_{max}}{m} \right)}\PKTs\psi_{max} \ge 0 \quad \text{for }r>0, m>0 \ppp \label{eq:A_D_eta2}
\end{equation}
In the derivation of equations \eqref{eq:A_D_eta1} and \eqref{eq:A_D_eta2} it was used that $\PKTs\psi_{max} \ne 0$ only when $\psi_0 = \psi_{max}$. Hence,
\begin{equation}
D = -\ABL\psi\von{\psi_0, \psi_{max}}\nach{\psi_{max}} \PKTs\psi_{max} = - \int\limits_0^{\psi_{0}} \ABL\eta\von{\xi, \psi_{max}}\nach{\psi_{max}} \mathrm{d}\xi \, \PKTs\psi_{max} = - \int\limits_0^{\psi_{max}} \ABL\eta\von{\xi, \psi_{max}}\nach{\psi_{max}} \mathrm{d}\xi \, \PKTs\psi_{max} \ppp
\end{equation}

\bibliographystyle{elsarticle-harv}
\bibliography{Literatur}

\begin{thebibliography}{37}
\expandafter\ifx\csname natexlab\endcsname\relax\def\natexlab#1{#1}\fi
\providecommand{\url}[1]{\texttt{#1}}
\providecommand{\href}[2]{#2}
\providecommand{\path}[1]{#1}
\providecommand{\DOIprefix}{doi:}
\providecommand{\ArXivprefix}{arXiv:}
\providecommand{\URLprefix}{URL: }
\providecommand{\Pubmedprefix}{pmid:}
\providecommand{\doi}[1]{\href{http://dx.doi.org/#1}{\path{#1}}}
\providecommand{\Pubmed}[1]{\href{pmid:#1}{\path{#1}}}
\providecommand{\bibinfo}[2]{#2}
\ifx\xfnm\relax \def\xfnm[#1]{\unskip,\space#1}\fi
%Type = Article
\bibitem[{Ball(1976)}]{Ball1976}
\bibinfo{author}{Ball, J.M.}, \bibinfo{year}{1976}.
\newblock \bibinfo{title}{{Convexity conditions and existence theorems in
  nonlinear elasticity}}.
\newblock \bibinfo{journal}{Archive for Rational Mechanics and Analysis}
  \bibinfo{volume}{63}, \bibinfo{pages}{337--403}.
%Type = Article
\bibitem[{Besdo and Ihlemann(2003)}]{Besdo2003a}
\bibinfo{author}{Besdo, D.}, \bibinfo{author}{Ihlemann, J.},
  \bibinfo{year}{2003}.
\newblock \bibinfo{title}{{Properties of rubberlike materials under large
  deformations explained by self-organizing linkage patterns}}.
\newblock \bibinfo{journal}{International Journal of Plasticity}
  \bibinfo{volume}{19}, \bibinfo{pages}{1001--1018}.
%Type = Article
\bibitem[{Bueche(1960)}]{Bueche1960}
\bibinfo{author}{Bueche, F.}, \bibinfo{year}{1960}.
\newblock \bibinfo{title}{{Molecular basis for the mullins effect}}.
\newblock \bibinfo{journal}{Journal of Applied Polymer Science}
  \bibinfo{volume}{4}, \bibinfo{pages}{107--114}.
%Type = Article
\bibitem[{Chagnon et~al.(2004)Chagnon, Verron, Gornet, Marckmann and
  Charrier}]{Chagnon2004}
\bibinfo{author}{Chagnon, G.}, \bibinfo{author}{Verron, E.},
  \bibinfo{author}{Gornet, L.}, \bibinfo{author}{Marckmann, G.},
  \bibinfo{author}{Charrier, P.}, \bibinfo{year}{2004}.
\newblock \bibinfo{title}{{On the relevance of Continuum Damage Mechanics as
  applied to the Mullins effect in elastomers}}.
\newblock \bibinfo{journal}{Journal of the Mechanics and Physics of Solids}
  \bibinfo{volume}{52}, \bibinfo{pages}{1627--1650}.
%Type = Article
\bibitem[{Coleman and Gurtin(1967)}]{Coleman1967a}
\bibinfo{author}{Coleman, B.}, \bibinfo{author}{Gurtin, M.},
  \bibinfo{year}{1967}.
\newblock \bibinfo{title}{{Thermodynamics with internal state variables}}.
\newblock \bibinfo{journal}{The Journal of Chemical Physics}
  \bibinfo{volume}{47}, \bibinfo{pages}{597}.
%Type = Article
\bibitem[{Coleman(1964)}]{Coleman1964a}
\bibinfo{author}{Coleman, B.D.}, \bibinfo{year}{1964}.
\newblock \bibinfo{title}{{Thermodynamics of materials with memory}}.
\newblock \bibinfo{journal}{Archive for Rational Mechanics and Analysis}
  \bibinfo{volume}{17}.
%Type = Article
\bibitem[{Coleman and Noll(1963)}]{Coleman1963a}
\bibinfo{author}{Coleman, B.D.}, \bibinfo{author}{Noll, W.},
  \bibinfo{year}{1963}.
\newblock \bibinfo{title}{{The thermodynamics of elastic materials with heat
  conduction and viscosity}}.
\newblock \bibinfo{journal}{Archive for Rational Mechanics and Analysis}
  \bibinfo{volume}{13}, \bibinfo{pages}{167--178}.
%Type = Article
\bibitem[{Diani et~al.(2009)Diani, Fayolle and Gilormini}]{Diani2009}
\bibinfo{author}{Diani, J.}, \bibinfo{author}{Fayolle, B.},
  \bibinfo{author}{Gilormini, P.}, \bibinfo{year}{2009}.
\newblock \bibinfo{title}{{A review on the Mullins effect}}.
\newblock \bibinfo{journal}{European Polymer Journal} \bibinfo{volume}{45},
  \bibinfo{pages}{601--612}.
%Type = Inproceedings
\bibitem[{Diercks and Lion(2013)}]{Diercks2013}
\bibinfo{author}{Diercks, N.}, \bibinfo{author}{Lion, A.},
  \bibinfo{year}{2013}.
\newblock \bibinfo{title}{{Modelling deformation-induced anisotropy using
  1D-laws for Mullins-Effect}}, in: \bibinfo{editor}{{N. Gil-Negrete}, A.}
  (Ed.), \bibinfo{booktitle}{Constitutive Models for Rubber VIII},
  \bibinfo{publisher}{CRC Press}. pp. \bibinfo{pages}{419--424}.
%Type = Article
\bibitem[{Dorfmann and Ogden(2003)}]{Dorfmann2003}
\bibinfo{author}{Dorfmann, A.}, \bibinfo{author}{Ogden, R.W.},
  \bibinfo{year}{2003}.
\newblock \bibinfo{title}{{A pseudo-elastic model for loading, partial
  unloading and reloading of particle-reinforced rubber}}.
\newblock \bibinfo{journal}{International Journal of Solids and Structures}
  \bibinfo{volume}{40}, \bibinfo{pages}{2699--2714}.
%Type = Article
\bibitem[{Dorfmann and Ogden(2004)}]{Dorfmann2004}
\bibinfo{author}{Dorfmann, A.}, \bibinfo{author}{Ogden, R.W.},
  \bibinfo{year}{2004}.
\newblock \bibinfo{title}{{A constitutive model for the Mullins effect with
  permanent set in particle-reinforced rubber}}.
\newblock \bibinfo{journal}{International Journal of Solids and Structures}
  \bibinfo{volume}{41}, \bibinfo{pages}{1855--1878}.
%Type = Article
\bibitem[{El\'{\i}as-Z\'{u}\~{n}iga and Beatty(2002)}]{Zuniga2002}
\bibinfo{author}{El\'{\i}as-Z\'{u}\~{n}iga, A.}, \bibinfo{author}{Beatty,
  M.F.}, \bibinfo{year}{2002}.
\newblock \bibinfo{title}{{A new phenomenological model for stress-softening in
  elastomers}}.
\newblock \bibinfo{journal}{Zeitschrift f\"{u}r angewandte Mathematik und
  Physik} \bibinfo{volume}{53}, \bibinfo{pages}{794--814}.
%Type = Article
\bibitem[{Freund and Ihlemann(2010)}]{Freund2010}
\bibinfo{author}{Freund, M.}, \bibinfo{author}{Ihlemann, J.},
  \bibinfo{year}{2010}.
\newblock \bibinfo{title}{{Generalization of one-dimensional material models
  for the finite element method}}.
\newblock \bibinfo{journal}{ZAMM - Journal of Applied Mathematics and
  Mechanics} \bibinfo{volume}{90}, \bibinfo{pages}{399--417}.
%Type = Article
\bibitem[{G\"{o}ktepe and Miehe(2005)}]{Goktepe2005}
\bibinfo{author}{G\"{o}ktepe, S.}, \bibinfo{author}{Miehe, C.},
  \bibinfo{year}{2005}.
\newblock \bibinfo{title}{{A micro-macro approach to rubber-like materials.
  Part III: The micro-sphere model of anisotropic Mullins-type damage}}.
\newblock \bibinfo{journal}{Journal of the Mechanics and Physics of Solids}
  \bibinfo{volume}{53}, \bibinfo{pages}{2259--2283}.
%Type = Article
\bibitem[{Govindjee and Simo(1991)}]{Govindjee1991}
\bibinfo{author}{Govindjee, S.}, \bibinfo{author}{Simo, J.C.},
  \bibinfo{year}{1991}.
\newblock \bibinfo{title}{{A micro-mechanically based continuum damage model
  for carbon black-filled rubbers incorporating Mullins' effect}}.
\newblock \bibinfo{journal}{Journal of the Mechanics and Physics of Solids}
  \bibinfo{volume}{39}, \bibinfo{pages}{87--112}.
%Type = Article
\bibitem[{Govindjee and Simo(1992)}]{Govindjee1992}
\bibinfo{author}{Govindjee, S.}, \bibinfo{author}{Simo, J.C.},
  \bibinfo{year}{1992}.
\newblock \bibinfo{title}{{Transition from micro-mechanics to computationally
  efficient phenomenology: Carbon black filled rubbers incorporating mullins'
  effect}}.
\newblock \bibinfo{journal}{Journal of the Mechanics and Physics of Solids}
  \bibinfo{volume}{40}, \bibinfo{pages}{213--233}.
%Type = Article
\bibitem[{Guo and Sluys(2006)}]{Guo2006}
\bibinfo{author}{Guo, Z.}, \bibinfo{author}{Sluys, L.J.}, \bibinfo{year}{2006}.
\newblock \bibinfo{title}{{Computational modelling of the stress-softening
  phenomenon of rubber-like materials under cyclic loading}}.
\newblock \bibinfo{journal}{European Journal of Mechanics - A/Solids}
  \bibinfo{volume}{25}, \bibinfo{pages}{877--896}.
%Type = Book
\bibitem[{Haupt(1999)}]{Haupt1999}
\bibinfo{author}{Haupt, P.}, \bibinfo{year}{1999}.
\newblock \bibinfo{title}{{Continuum Mechanics and Theory of Materials}}.
\newblock \bibinfo{publisher}{Springer}.
%Type = Article
\bibitem[{Itskov et~al.(2010)Itskov, Ehret, Kazakevic̆iutė-Makovska and
  Weinhold}]{Itskov2010}
\bibinfo{author}{Itskov, M.}, \bibinfo{author}{Ehret, A.E.},
  \bibinfo{author}{Kazakevic̆iutė-Makovska, R.}, \bibinfo{author}{Weinhold,
  G.}, \bibinfo{year}{2010}.
\newblock \bibinfo{title}{{A thermodynamically consistent phenomenological
  model of the anisotropic Mullins effect}}.
\newblock \bibinfo{journal}{ZAMM - Journal of Applied Mathematics and
  Mechanics} \bibinfo{volume}{90}, \bibinfo{pages}{370--386}.
%Type = Article
\bibitem[{Kaliske et~al.(2001)Kaliske, Nasdala and Rothert}]{Kaliske2001}
\bibinfo{author}{Kaliske, M.}, \bibinfo{author}{Nasdala, L.},
  \bibinfo{author}{Rothert, H.}, \bibinfo{year}{2001}.
\newblock \bibinfo{title}{{On damage modelling for elastic and viscoelastic
  materials at large strain}}.
\newblock \bibinfo{journal}{Computers \& Structures} \bibinfo{volume}{79},
  \bibinfo{pages}{2133--2141}.
%Type = Article
\bibitem[{Kazakevi\v{c}iute-Makovska(2007)}]{Kazakeviciute-Makovska2007}
\bibinfo{author}{Kazakevi\v{c}iute-Makovska, R.}, \bibinfo{year}{2007}.
\newblock \bibinfo{title}{{Experimentally determined properties of softening
  functions in pseudo-elastic models of the Mullins effect}}.
\newblock \bibinfo{journal}{International Journal of Solids and Structures}
  \bibinfo{volume}{44}, \bibinfo{pages}{4145--4157}.
%Type = Article
\bibitem[{Kl\"{u}ppel and Schramm(2000)}]{Kluppel2000}
\bibinfo{author}{Kl\"{u}ppel, M.}, \bibinfo{author}{Schramm, J.},
  \bibinfo{year}{2000}.
\newblock \bibinfo{title}{{A generalized tube model of rubber elasticity and
  stress softening of filler reinforced elastomer systems}}.
\newblock \bibinfo{journal}{Macromolecular Theory and Simulations}
  \bibinfo{volume}{9}, \bibinfo{pages}{742--754}.
%Type = Article
\bibitem[{Lion(1996)}]{Lion1996}
\bibinfo{author}{Lion, A.}, \bibinfo{year}{1996}.
\newblock \bibinfo{title}{{A constitutive model for carbon black filled rubber:
  Experimental investigations and mathematical representation}}.
\newblock \bibinfo{journal}{Continuum Mechanics and Thermodynamics}
  \bibinfo{volume}{8}, \bibinfo{pages}{153--169}.
%Type = Article
\bibitem[{Marckmann et~al.(2002)Marckmann, Verron, Gornet, Chagnon, Charrier
  and Fort}]{Marckmann2002}
\bibinfo{author}{Marckmann, G.}, \bibinfo{author}{Verron, E.},
  \bibinfo{author}{Gornet, L.}, \bibinfo{author}{Chagnon, G.},
  \bibinfo{author}{Charrier, P.}, \bibinfo{author}{Fort, P.},
  \bibinfo{year}{2002}.
\newblock \bibinfo{title}{{A theory of network alteration for the Mullins
  effect}}.
\newblock \bibinfo{journal}{Journal of the Mechanics and Physics of Solids}
  \bibinfo{volume}{50}, \bibinfo{pages}{2011--2028}.
%Type = Article
\bibitem[{Miehe(1995)}]{Miehe1995a}
\bibinfo{author}{Miehe, C.}, \bibinfo{year}{1995}.
\newblock \bibinfo{title}{{Discontinuous and continuous damage evolution in
  Ogden-type large-strain elastic materials}}.
\newblock \bibinfo{journal}{European Journal of Mechanics - A/Solids}
  \bibinfo{volume}{14}, \bibinfo{pages}{697 -- 720}.
%Type = Article
\bibitem[{Miehe and Keck(2000)}]{Miehe2000}
\bibinfo{author}{Miehe, C.}, \bibinfo{author}{Keck, J.}, \bibinfo{year}{2000}.
\newblock \bibinfo{title}{{Superimposed finite
  elastic-viscoelastic-plastoelastic stress response with damage in filled
  rubbery polymers. Experiments, modelling and algorithmic implementation}}.
\newblock \bibinfo{journal}{Journal of the Mechanics and Physics of Solids}
  \bibinfo{volume}{48}, \bibinfo{pages}{323--365}.
%Type = Article
\bibitem[{Mullins(1948)}]{Mullins1948}
\bibinfo{author}{Mullins, L.}, \bibinfo{year}{1948}.
\newblock \bibinfo{title}{{Effect of Stretching on the Properties of Rubber}}.
\newblock \bibinfo{journal}{Rubber Chemistry and Technology}
  \bibinfo{volume}{21}, \bibinfo{pages}{281--300}.
%Type = Article
\bibitem[{Ogden and Roxburgh(1999)}]{Ogden1999}
\bibinfo{author}{Ogden, R.}, \bibinfo{author}{Roxburgh, D.},
  \bibinfo{year}{1999}.
\newblock \bibinfo{title}{{A pseudo-elastic model for the Mullins effect in
  filled rubber}}.
\newblock \bibinfo{journal}{Proceedings of the Royal Society of London. Series
  A: Mathematical, Physical and Engineering Sciences} \bibinfo{volume}{455},
  \bibinfo{pages}{2861--2877}.
%Type = Article
\bibitem[{Ogden(2001)}]{Ogden2001}
\bibinfo{author}{Ogden, R.W.}, \bibinfo{year}{2001}.
\newblock \bibinfo{title}{{Stress softening and residual strain in the
  azimuthal shear of a pseudo-elastic circular cylindrical tube}}.
\newblock \bibinfo{journal}{International Journal of Non-Linear Mechanics}
  \bibinfo{volume}{36}, \bibinfo{pages}{477--487}.
%Type = Article
\bibitem[{Pe\~{n}a et~al.(2009)Pe\~{n}a, Pe\~{n}a and Doblar\'{e}}]{Pena2009}
\bibinfo{author}{Pe\~{n}a, E.}, \bibinfo{author}{Pe\~{n}a, J.A.},
  \bibinfo{author}{Doblar\'{e}, M.}, \bibinfo{year}{2009}.
\newblock \bibinfo{title}{{On the Mullins effect and hysteresis of fibered
  biological materials: A comparison between continuous and discontinuous
  damage models}}.
\newblock \bibinfo{journal}{International Journal of Solids and Structures}
  \bibinfo{volume}{46}, \bibinfo{pages}{1727--1735}.
%Type = Article
\bibitem[{Qi and Boyce(2004)}]{Qi2004}
\bibinfo{author}{Qi, H.}, \bibinfo{author}{Boyce, M.}, \bibinfo{year}{2004}.
\newblock \bibinfo{title}{{Constitutive model for stretch-induced softening of
  the stress-stretch behavior of elastomeric materials}}.
\newblock \bibinfo{journal}{Journal of the Mechanics and Physics of Solids}
  \bibinfo{volume}{52}, \bibinfo{pages}{2187--2205}.
%Type = Article
\bibitem[{Schr\"{o}der and Neff(2003)}]{Schroder2003}
\bibinfo{author}{Schr\"{o}der, J.}, \bibinfo{author}{Neff, P.},
  \bibinfo{year}{2003}.
\newblock \bibinfo{title}{{Invariant formulation of hyperelastic transverse
  isotropy based on polyconvex free energy functions}}.
\newblock \bibinfo{journal}{International Journal of Solids and Structures}
  \bibinfo{volume}{40}, \bibinfo{pages}{401--445}.
%Type = Article
\bibitem[{Shutov and Ihlemann(2011)}]{Shutov2011}
\bibinfo{author}{Shutov, A.V.}, \bibinfo{author}{Ihlemann, J.},
  \bibinfo{year}{2011}.
\newblock \bibinfo{title}{{On the simulation of plastic forming under
  consideration of thermal effects}}.
\newblock \bibinfo{journal}{Materialwissenschaft und Werkstofftechnik}
  \bibinfo{volume}{42}, \bibinfo{pages}{632--638}.
%Type = Article
\bibitem[{Shutov and Ihlemann(2014)}]{Shutov2014}
\bibinfo{author}{Shutov, A.V.}, \bibinfo{author}{Ihlemann, J.},
  \bibinfo{year}{2014}.
\newblock \bibinfo{title}{{Analysis of some basic approaches to finite strain
  elasto-plasticity in view of reference change}}.
\newblock \bibinfo{journal}{International Journal of Plasticity} .
%Type = Article
\bibitem[{Simo(1987)}]{Simo1987}
\bibinfo{author}{Simo, J.C.}, \bibinfo{year}{1987}.
\newblock \bibinfo{title}{{On a fully three-dimensional finite-strain
  viscoelastic damage model: Formulation and computational aspects}}.
\newblock \bibinfo{journal}{Computer Methods in Applied Mechanics and
  Engineering} \bibinfo{volume}{60}, \bibinfo{pages}{153--173}.
%Type = Article
\bibitem[{Wulf and Ihlemann(2011)}]{Wulf2011}
\bibinfo{author}{Wulf, H.}, \bibinfo{author}{Ihlemann, J.},
  \bibinfo{year}{2011}.
\newblock \bibinfo{title}{{Simulation of self-organization processes in filled
  rubber considering thermal agitation}}.
\newblock \bibinfo{journal}{Constitutive Models for Rubber VII}
  \bibinfo{volume}{VII}.
%Type = Article
\bibitem[{Zhang et~al.(2011)Zhang, Andrieux and Sun}]{Zhang2011}
\bibinfo{author}{Zhang, X.F.}, \bibinfo{author}{Andrieux, F.},
  \bibinfo{author}{Sun, D.Z.}, \bibinfo{year}{2011}.
\newblock \bibinfo{title}{{Pseudo-elastic description of polymeric foams at
  finite deformation with stress softening and residual strain effects}}.
\newblock \bibinfo{journal}{Materials \& Design} \bibinfo{volume}{32},
  \bibinfo{pages}{877--884}.

\end{thebibliography}




%% \bibitem must have one of the following forms:
%%   \bibitem[Jones et al.(1990)]{key}...
%%   \bibitem[Jones et al.(1990)Jones, Baker, and Williams]{key}...
%%   \bibitem[Jones et al., 1990]{key}...
%%   \bibitem[\protect\citeauthoryear{Jones, Baker, and Williams}{Jones
%%       et al.}{1990}]{key}...
%%   \bibitem[\protect\citeauthoryear{Jones et al.}{1990}]{key}...
%%   \bibitem[\protect\astroncite{Jones et al.}{1990}]{key}...
%%   \bibitem[\protect\citename{Jones et al., }1990]{key}...
%%   \harvarditem[Jones et al.]{Jones, Baker, and Williams}{1990}{key}...
%%


% \end{thebibliography}

\end{document}